\documentclass[12pt]{article}

\usepackage{amsmath,amssymb,fullpage}
\usepackage{algorithm}
\usepackage{algpseudocode}

\newtheorem{theorem}{Theorem}
\newtheorem{corollary}[theorem]{Corollary}
\newtheorem{lemma}[theorem]{Lemma}

\newtheorem{remark}[theorem]{Remark}
\newtheorem{claim}[theorem]{Claim}

\newtheorem{fact}[theorem]{Fact}
\newcommand{\poly}{{\mathrm{poly}}}
\newcommand{\nnz}{{\mathrm{nnz}}}

\newenvironment{proof}{\noindent {\bf Proof:}\ }{\qed \par\vskip 4mm\par}

\newcommand{\REAL}{\ensuremath{\mathbb{R}}}
\newcommand{\cost}{\ensuremath{\text{cost}}}
\newcommand{\opt}{\ensuremath{\text{opt}}}
\newcommand{\qed}{\hfill \ensuremath{\Box}}
\newcommand{\dist}{\ensuremath{\text{dist}}}

\begin{document}
\thispagestyle{empty}
\title{Constant Size Coresets for $k$-Median Clustering}

\author{Christian Sohler\thanks{TU Dortmund, Germany, christian.sohler@tu-dortmund.de }
			  \and David P. Woodruff \thanks{Carnegie Mellon University, USA, dwoodruf@cs.cmu.edu}}

\date{}
\title{Strong Coresets for $k$-Median and Subspace Approximation: Goodbye Dimension}
\maketitle

\begin{abstract}
We obtain the first strong coresets for the $k$-median and subspace
approximation problems with sum of distances objective function, 
on $n$ points in $d$ dimensions, with a number of weighted points that is
independent of both $n$ and $d$; namely, our coresets have size 
$\poly(k/\epsilon)$. A strong coreset $(1+\epsilon)$-approximates 
the cost function for 
all possible sets of centers simultaneously.
We also give efficient $\nnz(A) + (n+d)\poly(k/\epsilon) + \exp(\poly(k/\epsilon))$
time algorithms for computing these coresets. 

We obtain the result by introducing a new dimensionality reduction technique for coresets
that significantly generalizes an earlier result of Feldman, Sohler and Schmidt \cite{FSS13} for 
squared Euclidean distances to sums of $p$-th powers of Euclidean distances for  
constant $p\ge1$.
\end{abstract}

\thispagestyle{empty}
\clearpage
\setcounter{page}{1}
\section{Introduction}
Coresets are a technique for data size reduction, which have been developed for a large family of problems in machine learning and statistics.
Given a set $P$ of $n$ points $p_1, \ldots, p_n$ each in $\mathbb{R}^d$, loosely speaking a {\it coreset} is a low-memory data structure $D$ which 
can be used in place of $P$ to approximate the cost of any query $Q$ on $P$. For example, in the Frobenius norm subspace approximation problem, one 
may be interested in computing an approximation to $\sum_{i=1}^n \|p_i - p_i P_V\|_2^2$, where $P_V$ is the orthogonal projection onto a $k$-dimensional 
subspace $V$ which corresponds to the query $Q$. As another important example, in the $k$-means problem one may be given a query $Q = \{q_1, \ldots, q_k\}$ of $k$ points, 
and one may be interested in computing an approximation to $\sum_{i=1}^n \|p_i - n(p_i, Q)\|_2^2$, where $n(p_i, Q)$ denotes the closest point in $Q$ to point $p_i$. In these examples, the notion of approximation is a $(1+\epsilon)$-relative error approximation, that is, a value $(1 \pm \epsilon) \sum_{i=1}^n \|p_i - p_i P_V\|_2^2$ for the subspace approximation problem, and a value $(1 \pm \epsilon) \sum_{i=1}^n \|p_i - n(p_i, Q)\|_2^2$ for the $k$-means problem. 

Often in these problems one seeks a {\it strong coreset}, which means that with high probability, the data structure $D$ should work {\it simultaneously} for all queries $Q$. That is, one may use random choices in the construction of $D$, but after forming $D$ it should be the case that $D$ can be used to provide a $(1+\epsilon)$-relative error approximation for every possible query simultaneously. An advantage of such a coreset is that for {\it any objective function} for which a table of $(1+\epsilon)$-approximate values to all possible queries can be used to provide a $(1+O(\epsilon))$-approximation to the objective function, one can throw away the original set of points and instead just retain the data structure $D$. For example, note that the above coreset for subspace approximation contains enough information to approximately solve principal component analysis (PCA), since if one finds the query $k$-dimensional subspace with minimum approximate value, this provides a $k$-dimensional subspace providing a $(1+\epsilon)$-approximation to the space spanned by the top $k$ principal components. However, the above coreset for subspace approximation can also be used to solve the $k$-means problem, since the latter can be rewritten as a {\it constrained} low rank approximation problem \cite{BZMD15,CEMMP15}. Thus, given that a strong coreset approximately preserves the cost of any query, it can be used in place of the original point set in any application which depends only on the answers to the queries. Note that if the coreset were instead to only approximately preserve the cost of any fixed query with high probability, then it might not be possible to solve the problem using the coreset since one may need to adaptively query the data structure, and outputs to successive queries may no longer be correct since the inputs depend on outputs to previous queries. 

Another advantage of a coreset is if it small, then it leads to considerable efficiency gains. For example, in distributed settings, each machine which has a subset of input points can compress its input points to a coreset, and then communicate the coreset to a central coordinator. The central coordinator, who often has more resources available, can then combine the coresets and use them to optimize the desired function. As communication is a bottleneck, a small coreset gives rise to more efficient protocols. Similarly, when processing a data stream, a common technique is the merge-and-reduce framework, in which one partitions the stream into chunks, computes a coreset on each chunk, and merges the coresets in a binary tree like structure as one processes successive chunks of the data stream. A small coreset thus leads to small space streaming algorithms. 

A long line of work has focused on developing strong coresets for both the subspace approximation problem \cite{DRVW06,DV07,SV07,FMSW10,FL11,VX12,FSS13} and the $k$-means problem \cite{BHI02,HM04,FS05,FS06,HK07,C09,LS10,FL11,FS12,FSS13}. Prior to the work of \cite{FSS13}, all previous coresets stored a weighted set of points, and the query just consisted of evaluating the same objective function on these weighted points. Moreover all such works required storing a number of points that was at least $d$, and an important question was to obtain coresets with a number of points {\it independent} of $d$. In \cite{FSS13}, by taking the top 
$O(k/\epsilon)$ principal components of the input points, arranged as an $n \times d$ matrix $A$, it was shown how to obtain the first strong coresets for the subspace approximation problem with a number of points independent of $d$ and $n$, namely, the authors achieved a coreset size of $O(k/\epsilon)$ points. An important idea
to obtain this result was that the cost of projecting the points on the first $O(k/\epsilon)$ principal components is approximately present for every 
candidate subspace and therefore can be dealt with as an additive constant. The authors also extend this result to the $k$-means problem by proving that the
projection on the first $O(k/\epsilon^2)$ principal components together with an appropriate constant will provide a coreset (of linear size but smaller dimension)
for the $k$-means problem. Combining this with existing constructions they achieved a coreset size of $\poly(k/\epsilon)$ points.

The $O(k/\epsilon^2)$ bound for the $k$-means problem was improved in \cite{CEMMP15} by using the fact that the $k$-means problem can be viewed as a constrained subspace approximation problem. In \cite{CEMMP15} the authors also find such a coreset in $\nnz(A)$ time, where $\nnz(A)$ is the number of non-zero entries of $A$. 

A major open question was if one could obtain strong coresets {\it independent of $d$ (and $n$)} for {\it $k$-median} and the subspace approximation problem with {\it sum of distances} $\sum_{i=1}^n \|p_i - p_iP_V\|_2$, as opposed to the sum of squares of distances. Unlike the $k$-means and sum of squares objective for subspace approximation, the $k$-median and sum of distances measures are much less amenable to algebraic manipulation; indeed there is no singular value decomposition (SVD) which was the driving force behind previous results. Notably, this version of the subspace problem is NP-hard \cite{cw15}, unlike minimizing the sum of squares. 

\subsection{Our Contributions}
Our main result is the construction of the first strong coresets independent of the dimension $d$ and number $n$ of input points for the $k$-median problem, as well as for the subspace approximation problem with sum of distances $\sum_{i=1}^n \|p_i -p_i P_V\|_2$. Our coresets have size $\poly(k/\epsilon)$ for both problems, and consist of a weighted set of points with a small twist. We add a single extra dimension to each of our points! We explain this more below. 

Our main new technique is a dimensionality reduction that generalizes a result of \cite{FSS13} for sum of squared distances to $p$-th powers
of distances for any constant $p \ge 1$. We also show how to build a strong coreset for subspace approximation
with $p$-th powers of Euclidean distance cost measure, for constant $p \ge 1$. 
Finally, we show how to find such coresets in time $\tilde{O}(\nnz(A) + (n+d)\poly(k/\epsilon)) + \exp(\poly(k/\epsilon))$ 
for the $k$-median problem and for the subspace approximation problem with
$p \in [1, 2)$, and in $\nnz(A)\poly(k/\epsilon) +(n+d) \poly(k/\epsilon) + \exp(\poly(k/\epsilon))$ time for the subspace approximation problem with $p > 2$. 

\subsubsection{Dimensionality Reduction}
We start by outlining our dimensionality reduction technique for the sum of distances objective function. A natural approach to try is to find a low dimensional subspace $S$ of $\mathbb{R}^d$ so that for any rank-$k$ subspace $V$, $\|A-AP_V\|_{1,2} \approx \|B-BP_V\|_{1,2} + \|A-B\|_{1,2}$, where for a matrix $C$, $\|C\|_{1,2}$ denotes the sum of Euclidean norms of rows of $C$, and here $B$ is the projection of the rows of $A$ (our initial points) onto the subspace $S$. Indeed, this is exactly the approach taken by \cite{FSS13,CEMMP15} for the subspace approximation problem with sum of squares of distances, where among other constructions, $S$ can be chosen to be the span of the top $O(k/\epsilon)$ singular vectors of $A$. It can be shown that the sum of squared distances to any object that is contained in a $k$-dimensional subspace is roughly the
projection cost onto the optimal $O(k/\epsilon)$-dimensional subspace plus the cost of the projected points. One way to think of this approach is to split
the cost into a structured part (the low dimensional point set) and a ``pseudorandom'' part (the projection cost), where the pseudorandom part essentially
acts like a random point set, as its cost will occur for any $k$-dimensional object, while the structured part is the one that can differ. 
What significantly helps in the case of squared distances is the Pythagorean theorem, which often allows to easily express distances as the sum of 
``independent'' distances. For example, if our object is contained in the optimal $O(k/\epsilon)$-dimensional subspace, then the cost of each point is
the squared distance of the projection plus the squared distance from the projected point to the object. 

Unfortunately, we do not have such a simple formula for exponents other than $2$. For the sum of distances one can show that an analogous approach does not work.
Consider a set of $n$ points in $\REAL^d$ where each coordinate of each point is drawn independently from a Gaussian distribution with expectation $0$ and variance $1/d$, i.e., the expected squared length of each point is $1$. We will assume that $d$ is large, which implies that the squared length is sharply concentrated and 
the expected length of the vector is close to $1$. Assume now that similarly to the case of squared distances we project our input point set to a low dimensional
subspace that minimizes the sum of squared projection lengths and we would like to use the projected point set together with the sum of projection lengths as a coreset
(where the sum is used as an additive constant in the costs). We will now argue that this cannot work for sufficiently large $n$ and $d$. Assume that the dimension
of the low dimensional subspace is $\ell$, a value independent of $n$ and $d$. In order to understand the properties of the optimal subspace, we first consider an
arbitrary fixed subspace of dimension $\ell$. If we now consider a random vector $x=(x_1,\dots,x_d)$ where each $x_i$ is chosen from the Gaussian distribution 
as described above, we notice that since the Gaussian distribution is invariant under rotation, that the expected squared length of the projected
point is equal to the expected squared length of the random vector $x'=(x_1,\dots,x_{\ell}, 0,\dots, 0)$, which is $\ell/d$  and the expected squared length of 
the projection is $1-\ell/d$. For sufficiently large $d$ this length approaches $0$ and the expected squared length of the projection approaches $1$.
Thus, for a fixed subspace, the expected sum of squared distances is $n$ and for $n\rightarrow \infty$ we get that the sum of squared projection lengths is sharply concentrated. 
Using a union bound over a net of all subspaces we conclude that every $\ell$-dimensional subspace will have cost roughly $n$.

Now consider the cost of an arbitrary point $q$ at distance $1$ from the origin. The expected distance of an input point to this point is
roughly $\sqrt{2}$ and so the sum of distances will approach $\sqrt{2}n$ as $n \rightarrow \infty$. Now recall that the length of the projection
of the input points goes to $0$. Thus their distance to $q$ will be roughly $1$. Thus, the sum of distances of the projected points is roughly $n$ plus the projection cost, which is roughly $n$, and thus will give an estimate of $2n$, which is not a $(1+\epsilon)$-approximation. Hence, we cannot simply work with a single additive weight as in the case of squared distances.

Instead, we proceed as follows. We first describe the existence of a coreset and then how to find
it efficiently. We start with a $k$-dimensional subspace $S$ of $\mathbb{R}^d$ for which 
$\|A(I-P_S)\|_{1,2} = \min_{\textrm{rank-}k \textrm{ subspaces S}}\|A(I-P_{S})\|_{1,2} = \opt$, where $P_S$ denotes the 
orthogonal projection onto $S$. We iteratively augment $S$ by $k$-dimensional subspaces until the cost no longer drops
by $\epsilon^2 \opt$. That is, in the first step, we try to find a $2k$-dimensional subspace $S'$
containing $S$ for which $\|A(I-P_{S'})\|_{1,2} \leq \opt - \epsilon^2 \opt$. We then replace $S$ with $S'$.
In the second step, we try to find a $3k$-dimensional subspace $S'$ containing $S$ for which
$\|A(I-P_{S'})\|_{1,2} \leq \opt' - \epsilon^2 \opt$, 
where $\|A(I-P_{S})\|_{1,2} = \opt' \leq \opt - \epsilon^2 \opt$, etc. This process repeats for at most 
$\epsilon^{-2}$ steps, at which point we have an at most $k/\epsilon^2$-dimensional subspace $S$ for which for
any $k$-dimensional subspace $V$, $\|A(I-P_{V \cup S})\|_{1,2} \leq (1+\epsilon)\|A(I-P_S)\|_{1,2}$. The latter property 
can be shown to imply that $\|AP_{V \cup S} - AP_{S}\|_{1,2} \leq \epsilon \cdot \opt$, that is, the projections of the 
$n$ points onto $S$ are close to the projections of the $n$ points onto $V \cup S$, {\it for any} $V$. See 
Lemma \ref{lem:proj}. 

Next, since we can ``move'' each of the rows of $AP_{V \cup S}$ to the corresponding
rows in $AP_S$ by paying a total sum of distances cost of $\epsilon \cdot \opt$, it follows by the triangle inequality that
for any set $C$ of points that is contained in a $k$-dimensional subspace $V$, the sum of distances from the rows of $AP_S$ 
to their corresponding closest points in $C$ is within $\epsilon \cdot \opt$ of the sum of distances from the rows of $AP_{V \cup S}$ 
to their corresponding closest points in $C$. 

Now we want to replace our original points (the rows of $A$) with their projections
onto $S$, namely, replace $A$ with $AP_S$. Although this step by itself does not reduce the number $n$ of points,
each of the $n$ points after projection lives in a much lower $k/\epsilon^2$ dimensional subspace (rather than the
initial space which has dimension $d$), and we will then be able to apply coreset construction techniques which
depend on this much smaller dimension. For any set $C$ of points contained in a $k$-dimensional subspace $V$, by the Pythagorean theorem 
we can write the distance of a row $p$  
of $A$ to $C$ as $\sqrt{a^2 + b^2}$, where $a$ is the distance of $p$ to $V \cup S$, and $b$ is the distance
of the projection of $p$ onto $V \cup S$, to $C$. We instead try to approximate $\sqrt{a^2+b^2}$ by $\sqrt{f^2+g^2}$,
where $f$ is the distance of $p$ to $S$ and $g$ is the distance of the projection of $p$ onto $S$, to $C$. We observe 
in Lemma \ref{lem:numbers} that $|\sqrt{a^2+b^2} - \sqrt{f^2+g^2}| \leq |a-f| + |b-g|$, and we know that the
average values (over the $n$ points) of $|a-f|$ and $|b-g|$ are small by Lemma \ref{lem:proj} combined with the triangle inequality.
%

Note that the value $f$, which is different for each of the $n$ rows of $A$, 
{\it does not depend} on $V$ or $C$, whereas the value $g$ is exactly the distance of the row of 
$AP_S$ to $C$. If we were simply to define $B = AP_S$, then $\|B - B_C\|_{1,2}$, where the i-$th$ row of $B_C$ contains
the closest point (of the closure) of $C$ to the corresponding row of $B$, would fail to capture
the distances of the $n$ rows of $A$ to $S$. Further, unlike for
the $\|\cdot\|_{2,2}$ norm, we cannot add a single number $\|A(I-P_S)\|_{2,2}$ to account for this,
which is a technique used in \cite{FSS13,CEMMP15}; 
this is precisely the difficulty of the $\|\cdot\|_{1,2}$ norm that we must deal with. Instead, a crucial
idea is to {\it append one additional coordinate} to each row of $B$, where in the $i$-th row we append
$\|A_{i*}(I-P_S)\|_2$, where $P_S$ is the orthogonal projection onto $S$. Then, to compute the distance to a 
$k$-dimensional subspace $V$, instead of approximating $\|A-A_C\|_{1,2}$ by $\|B - B_C\|_{1,2}$, where the $i$-th rows 
of $A_C$ and $B_C$ contain the closest points (of the closure) of $C$ to the corresponding row of $A$ and the first $d$ coordinates of the 
corresponding row of $B$, respectively, 
we approximate $\|A-AP_V\|_{1,2}$ by $\|B - B_C I^T\|_{1,2}$, where $B_C I^T$ is the matrix which appends an all $0$ column to $B_C$.
 Thus, the norm of the $i$-th row of $B-B_C I^T$ is $\sqrt{f^2+g^2}$, where $f$ is the distance of $A_i$ to $S$,
captured by the $(d+1)$-st coordinate of $B$, and $g$ is the distance of $A_iP_S$ to $C$. Thus, we have ``encoded'' the distances
of the rows of $A$ to $AP_S$ in the coreset this way. Note that this appended additional coordinate cannot be taken out of each
row and combined into a single number, as in \cite{FSS13,CEMMP15}, because for each row, 
its square is added to the squared distance of a point to its projection onto $S$, and then
a square root is taken, so it occurs ``under the square root'' in the distance computations. 

{\it Optimizing the Running Time.} We next implement the steps above in $\tilde{O}(\nnz(A)) + (n+d)\poly(k/\epsilon) + \exp(\poly(k/\epsilon))$ time. 
We first show how to reformulate our algorithm in Section \ref{sec:alternative} so that it suffices to run any algorithm for finding an $ik$-dimensional subspace $S$ of $\mathbb{R}^d$ which a $(1+\epsilon^2/2)$-approximate $ik$-dimensional subspace with respect to the $\|\cdot\|_{1,2}$ norm, for a random integer $i$ in $\{1, 2, \ldots, 10/\epsilon^2\}$. It is known how to find such a subspace in $\nnz(A) + (n+d)\poly(k/\epsilon) + \exp(\poly(k/\epsilon))$ time \cite{cw15}. This is done in our {\sc DimensionalityReductionII} algorithm. 

After finding such a subspace $S$, which is an $ik$-dimensional subspace of $\mathbb{R}^d$, we need the distance of each row of $A$ to $S$. To do so, we set up $n$ regression problems, the $i$-th being:
$\min_x \|A_i - xV^T\|_2$, where $P_S = VV^T$ and $V^T$ is an orthonormal basis for $S$, which can be computed in $d \cdot \poly(k/\epsilon)$ time. Using input sparsity time algorithms for regression \cite{cw13,mm13,nn13}, if we choose a CountSketch matrix $S$ with 
$\poly(k/\epsilon)$ rows, then with probability at least $9/10$, 
$\|A_i S - xV^TS\|_2 = (1 \pm \epsilon)\|A_i - xV^T\|_2$ for all $x$. We can compute $AS$ in $\nnz(A)$ time and $V^TS$ in $d \cdot \poly(k/\epsilon)$ time, at which point we can solve the $n$ regression problems in $n \cdot \poly(k/\epsilon)$ time, so $\nnz(A) + (n+d)\poly(k/\epsilon)$ total time. We repeat the entire procedure $O(\log n)$ times, and take the median of our $O(\log n)$ estimates, for each $i \in \{1, 2, \ldots, n\}$, giving us $O(\nnz(A) \log n + (n+d)\poly(k/\epsilon) \log n) = \tilde{O}(\nnz(A) + (n+d) \poly(k/\epsilon))$ total time to obtain $(1+\epsilon)$-approximations of the distances for each row of $A$ to $S$. Here, 
for a function $f$, $\tilde{O}(f) = f \cdot \poly(\log f)$. We note that such approximations, when used as the last coordinates of the rows of matrix $B$, only change $\|B-BIPI^T\|_{1,2}$ by a $(1+\epsilon)$-factor, where $P$ is any rank-$k$ orthogonal projection matrix. See Lemma \ref{lem:regression2}. 

We also need to project each of the rows of $A$ onto $S$, which would take
more than $\nnz(A)$ time; indeed, just writing down such projections could take $\Omega(nd)$ time. Fortunately, for our coreset
constructions we describe next, we never need to explicitly perform this projection. We first show how to 
avoid this for the subspace approximation problem.

\subsubsection{Coreset Construction for Subspace Approximation}
We first explain the construction for $p = 1$, then how to optimize the running time. At this point we have a rank-$\poly(k/\epsilon)$ 
matrix $B$ (or more precisely a rank-$(ik+1)$ matrix $B$ for some $i = O(1/\epsilon^2)$) 
for which 
for all $k$-dimensional orthogonal projection matrices $P$, 
$\big| \|A-AP\|_{1,2} - \|B-BIPI^T\|_{1,2} \big| \le \epsilon \|A-AP\|_{1,2}.$ 
Unfortunately the dimension of $B$ is still $n \times d$. Viewing the rows as points in $\mathbb{R}^d$, we would like to 
reduce the number $n$ of points. We first observe that since $\ell_2$ embeds linearly into $\ell_1$ with distortion at 
most $(1+\epsilon)$ via multiplication by a Gaussian matrix $G$ (this is a special case Dvoretsky's theorem), we have
$\|BG-BIPI^TG\|_{1,1} = (1 \pm \epsilon) \|B-BIPI^T\|_{1,2}$ for all $k$-dimensional orthogonal projection matrices $P$. Here $\|\cdot\|_{1,1}$
denotes the entrywise $1$-norm of a matrix. The intuition here is that $G$ maps $\mathbb{R}^d$ to $\mathbb{R}^{d'}$ for a $d' > d$
for which the image of $\mathbb{R}^d$ in $\mathbb{R}^{d'}$ consists of only flat vectors, so the $\ell_1$-norm of every vector
coincides with its $\ell_2$-norm, up to $(1 \pm \epsilon)$-factor, after scaling by the square root of the dimension. Here, 
$d' = O(d (\log(1/\epsilon))/\epsilon^2)$, but it does not matter for our purposes since we will never actually
instantiate $G$.

Next, we appeal to the Lewis weight sampling result of \cite{cp15}, which says that one can find a sampling and rescaling matrix
$T$ (a matrix which just samples rows and rescales them by positive weights) 
with $O(r \log(r)/\epsilon^2)$ rows for which for any rank-$r$ space $C$, $\|TCx\|_1 = (1\pm \epsilon)\|Cx\|_1$ simultaneously
for all $x$. Noting that for every $P$, each column of $BG-BIPI^TG$ is in the column span of $B$, which is a $\poly(k/\epsilon)$-dimensional subspace, we have $\|TBG-TBIPI^TG\|_{1,1} = (1 \pm \epsilon)\|BG-BIPI^TG\|_{1,1}$ for all rank-$k$ orthogonal projection matrices $P$, where
$T$ has $\poly(k/\epsilon)$ rows. Finally, noting that the {\it rows} of $TBG-TBIPI^TG$ are still in the row span of $G$, we can
apply Dvoretsky's theorem one more time to conclude that $\|TBG-TBIPI^TG\|_{1,1} = (1 \pm \epsilon)\|TB-TBIPI^T\|_{1,2}$ for all
rank-$k$ orthogonal projection matrices $P$. Stringing the inequalities together, we obtain $\|TB-TBIPI^T\|_{1,2} = (1 \pm \Theta(\epsilon))\|B-BIPI^T\|_{1,2}$ for all rank-$k$ orthogonal projection matrices. Note that we never need to multiply by $G$. Rather $G$ is a tool in the analysis
which shows the sampling procedure of \cite{cp15} works for sums of Euclidean norms. 

Consequently, our strong coreset consists of the rows of $TB$, so $\poly(k/\epsilon)$ points in $\mathbb{R}^d$, i.e., the rows
of $T$. These are the analogue of the $k/\epsilon$ right singular vectors of \cite{FSS13,CEMMP15} used to obtain a strong coreset for
the $\|\cdot\|_{2,2}$ error measure. We discuss a similar argument for $p$-th powers below.

{\it Optimizing the Running Time.} We now obtain $\tilde{O}(\nnz(A) + (n+d)\poly(k/\epsilon)) + \exp(\poly(k/\epsilon))$ running time. 
In this running time we can find the $\poly(k/\epsilon)$-dimensional subspace $S$ for which 
$B = [AP_S, v]$, where $P_S$ is the orthogonal projection
onto the $\poly(k/\epsilon)$-dimensional subspace found by our {\sc DimensionalityReductionII} algorithm, and 
$v_i = (1 \pm \epsilon) \|A_{i*} - A'_{i*}\|_{2}$,
where $A' = AP_S$ for $i = 1, 2, \ldots, n$.

As above, let $V \in \mathbb{R}^{d \times \poly(k/\epsilon)}$ have columns which form an orthonormal basis for the column span of $S$. 
To find the sampling and rescaling matrix $T$, the procedure in Theorem 1.1 of \cite{cp15} takes time equal to that of 
$O(\log \log n)$ invocations of
constant factor $\ell_2$-leverage score approximations of matrices of the form $W A V$, where $W$ is a non-negative 
diagonal matrix. We use the input sparsity time approximate leverage score samplers of \cite{cw13,mm13,nn13}, which compute
$S W A$ for a CountSketch matrix $S$ with $\poly(k/\epsilon)$ rows. This procedure computes $SWA$ in $O(\nnz(A))$ time,
then computes $SWAV$ in $d \cdot \poly(k/\epsilon)$ time, then a QR factorization in $\poly(k/\epsilon)$ time, then 
$(WAV)(R^{-1}G)$ for a Gaussian matrix $G$ with $O(\log n)$ columns. The row norms of $WAV (R^{-1}G)$ can be computed
in $\nnz(A) \log n + d \cdot \poly(k/\epsilon) \log n$ time using that $G$ has only $O(\log n)$ columns. Since the procedure
reduces to $O(\log \log n)$ invocations of this, in total this gives $\tilde{O}(\nnz(A) + (n+d)\poly(k/\epsilon))$ time to
find the matrix $T$. Finally, $T$ selects $\poly(k/\epsilon)$ rows of $A$, and for each we compute its projection onto $S$,
taking $d \cdot \poly(k/\epsilon)$ time in total. We also output the corresponding entry of $v$. We thus obtain our coreset
$TB$ in $\tilde{O}(\nnz(A) + (n+d)\poly(k/\epsilon))$ total time. 

Our coresets for subspace approximation with sum of $p$-th powers error measure follows via similar techniques. The
running time is slightly worse for $p > 2$ due to the fact that we can only implement our {\sc DimensionalityReductionII}
algorithm in $\tilde{O}(\nnz(A) + (n+d)\poly(k/\epsilon)) + \exp(\poly(k/\epsilon))$ time if $p \in [1, 2)$. For
$p > 2$ we use a slower algorithm running in $O(\nnz(A)\poly(k/\epsilon) + \exp(\poly(k/\epsilon))$ due to \cite{DV07} (they state their algorithm as $O(nd\poly(k/\epsilon) + \exp(\poly(k/\epsilon)))$ but if $A$ is sparse, the $nd \cdot \poly(k/\epsilon)$ can be replaced with an $\nnz(A) \cdot \poly(k/\epsilon)$ given that their algorithm just requires computing projections). 

\subsubsection{Coreset Construction for $k$-Median}

To obtain a coreset for $k$-median, we first apply our dimensionality reduction to get a matrix $B$ such that 
for every set of $k$-centers $C$ we have $\big| \|A-A_C\|_{1,2} - \|B-B_C I^T\|_{1,2} \big| \le \epsilon \cdot \|A-A_C\|_{1,2}$,
where $A_C$ and $B_C$ denote the matrices that contain in the $i$-th row the closest center of $C$ to the $i$-th row of $A$ and
$B$, respectively. We note that $B$ can be viewed as a point set in $O(k/\epsilon^2)$ dimensions. We can then use an arbitrary
coreset construction for this low dimensional point set where we append $k$ arbitrary dimensions to the space. Thus, the effect
of the construction will be to replace the $d$ in a coreset construction by $O(k/\epsilon^2)$.
We claim that 
a coreset for this enlarged space is also a coreset for the $d$-dimensional space. The reason is that any set of $k$-centers
in the $d$-dimensional space is either in the span of $B$ (in which case the coreset guarantee holds) or there is an orthogonal
transformation that does not change $B$ and maps the remaining centers to the $k$ added dimensions. This implies that
the coreset property holds for the full space. Thus, the cost of the coreset approximates the cost of $B$ upto a factor of
$1\pm\epsilon$. Combining this with the error bound of $\big| \|A-A_C\|_{1,2} - \|B-B_C\|_{1,2} \big| \le \epsilon \cdot \|A-A_C\|_{1,2}$
gives that the resulting set will be a $1+O(\epsilon)$ coreset and the result follows by rescaling $\epsilon$ by a constant.
Notice that the guarantee the coreset provides is slightly stronger than what we need as our centers 
will always have the last (special) coordinate equal to $0$. 

Plugging in the $k$-median coresets of \cite{FL11} or \cite{BFL16}, which are both of size $O(\frac{dk\log k}{\epsilon^2})$ 
(the first one has negative weights, which may be undesirable in some situations), we obtain a coreset of size $O(\frac{k^2\log k}{\epsilon^4})$.

In order to get a running time of $O(\nnz(A) + (n+d) \poly(k/\epsilon) + \exp(\poly(k/\epsilon)))$ we approximate the matrix $B$ of projections with a factored low rank matrix of approximate projections, see Lemma \ref{lemma:tildeB}. 


\subsubsection{Outline} In Section \ref{sec:prelim}, we give preliminaries.
In Section \ref{sec:dim}, we provide our main dimensionality reduction
technique. In Section \ref{sec:subspace}, we obtain our coresets
for subspace approximation. Finally, in Section \ref{sec:kmedian}, we obtain
our coreset for $k$-median. 

\section{Preliminaries}\label{sec:prelim}

We use $A \in \REAL^{n\times d}$ to denote a point set of $n$ points in $d$ dimensions (the rows of $A$). 
$A_{i*}$ denotes the $i$-th row of $A$ and $A_{*j}$ denotes the $j$-th column. For a matrix $A \in \REAL^{n\times d}$ we use $\|A\|_{p,2} = 
(\sum_{i=1}^n \|A_{i*}\|_2^p)^{1/p}$. In particular, we have $\|A\|_{2,2}^2 = \|A\|_F^2$, where $\|A\|_F$ denotes the Frobenius norm of
$A$. For a subspace $S$ we use $\cost_p(A,S)$ to denote the sum of $p$-th powers of the $l_2$-distances from the rows of $A$ to $S$.
For a non-empty set of points $C \subseteq \REAL^d$ we define $\dist(p,C) = \inf_{q\in C} \|p-q\|_2$.

We start with a few claims that will be useful to deal with norms and powers of norms. These are elementary properties
about numbers and we defer the proofs to the Appendix.

\begin{claim}
\label{claim:one}
Let $a,b,c \ge 0$ such that $a^2=b^2-c^2$. For $p \ge 2$ we have $a^p \le b^p - c^p$.
\end{claim}

\begin{claim}
\label{claim:two}
Let $a,b,c \ge 0$ such that $a^2=b^2-c^2$, $a^p \ge \epsilon b^p$ and $b^p \ge c^p$. Let $1\le p \le 2$ and $1\ge \epsilon >0$.
Then $a^p \le 10 \cdot  \epsilon^{\frac{p-2}{p}} \cdot (b^p-c^p)$.
\end{claim}

%

\begin{claim}
\label{claim:four}
Let $a,b,x \ge 0$. Let $1\ge p >0$. Then
$$
\big| (a+x)^p - (b+x)^p\big| \le \big| a^p - b^p\big|.
$$
\end{claim}

%

\begin{lemma}\label{lem:numbers}
Let $a,b,f,g \ge 0$. Then we have
$$
|\sqrt{a^2+b^2} - \sqrt{f^2+g^2}| \le |a-f| +|b-g|.
$$
\end{lemma}

%
%

\begin{claim}
\label{claim:five}
Let $a,b \ge 0$ and $1\ge \epsilon >0$ and $p \ge 1$. Then
$$
(a+b)^p \le (1+\epsilon) a^p + (1+\frac{2p}{\epsilon})^p b^p.
$$
\end{claim}

\section{Dimensionality Reduction}\label{sec:dim}

Our first result is a dimensionality reduction lemma for clustering problems where the cluster centers are
contained in a low-dimensional subspace such as, for example, $k$-median clustering. 

\begin{algorithm}[H]
	\caption{Dimensionality Reduction Algorithm}
	\begin{algorithmic}[1]
		\Procedure{DimensionalityReduction}{$A,n,d,k, \epsilon, p$}	
		\State Compute a $(1+\epsilon)$-approximation $S$ to the $k$-subspace problem with cost function sum of $p$-th powers of $l_2$-distances
		\State Let $\opt$ denote the cost of an optimal solution to the above problem
		\State Let $k^*=k$
		\While{there exists a subspace $S' \supseteq S$ of rank $k^*+k$ such that $\cost_p(A,S') \le \cost_p(A,S) - \epsilon^{\max\{\frac{2}{p},1\}} \opt/80$} 
		\State $k^*=k^*+k$
		\State $S=S'$
		\EndWhile
		\State Let $A'$ be the projection of $A$ on $S$
		\State Let $B \in \REAL^{n \times (d+1)}$ be a matrix whose entry at position $1\le i,j\le d$ equals the entry of $A'$
		and whose entries in the last column are $\|A_{i*} - A'_{i*}\|_{2}$
		\State \Return $B$
		\EndProcedure
\end{algorithmic}
\end{algorithm}
                                             
In the next lemma we show for the output space $S$ of dimension $\ell$ of the above algorithm and any subspace $S^*$ of
dimension $\ell+k$ that contains $S$ that the corresponding projections of the rows of $A$ onto $S$ and $S^*$ have small
distance on average. 

\begin{lemma}\label{lem:proj}
Let $1\ge \epsilon >0$ and $p \ge 1$.
Let $A$ be the input matrix of algorithm {\sc DimensionalityReduction}. Let $\opt$ be the cost of an optimal solution to the
linear $k$-subspace problem with respect to the sum of $p$-th powers of $l_2$-distances.
Let $S\subseteq \REAL^d$ be 
the subspace in the last iteration of the while loop and let $\ell$ be its dimension. Let $S^*\subseteq \REAL^d$ 
be an arbitrary subspace of dimension $k+\ell$ that contains $S$. Let $P$ and $P^*$ be orthogonal projection matrices onto $S$ and $S^*$. 
Then we have
$$
 \|AP-AP^*\|_{p,2}^p \le \epsilon \opt.
$$
\end{lemma}

\begin{proof}
We know from the algorithm that $\|A-AP\|_{p,2}^p - \|A-AP^*\|_ {p,2}^p \le \epsilon^{\max\{\frac{2}{p},1\}} \opt/80$. 
Furthermore, we have $\|A-AP\|_{p,2}^p \le (1+\epsilon) \cdot \opt$ by the way $S$ is computed. 
We first consider the case when $p = 2$. 
Since $A_{i*}- A_{i*} P^*$ is orthogonal to $A_{i*}P^*-A_{i*}P$ we know that in this case
\begin{eqnarray*}
\label{eqn:one}
\|A_{i*}P-A_{i*}P^*\|_2^2 & =&  \|A_{i*}-A_{i*}P\|_2^2 - \|A_{i*}-A_{i*}P^*\|_2^2 
\end{eqnarray*}
Applying the above equality row wise we obtain
$$
\|AP-AP^*\|_{2,2}^2 = \sum_{i=1}^n \|A_{i*}P-A_{i*}P^*\|_{2}^2 =  \sum_{i=1}^n ( \|A_{i*}-A_{i*}P\|_{2}^2 - \|A_{i*}-A_{i*}P^*\|_{2}^2) \le \epsilon \cdot \opt.
$$ 
Next we consider $p>2$. We define $a= \|A_{i*}P-A_{i*}P^*\|_2$, $b=\|A_{i*}-A_{i*}P\|_2$ and $c=\|A_{i*}-A_{i*}P^*\|_2$.
We observe that $a^2 = b^2 -c^2$ and so by Claim \ref{claim:one} we obtain that
$$
\|A_{i*}P-A_{i*}P^*\|_2^p  \le  \|A_{i*}-A_{i*}P\|_2^p - \|A_{i*}-A_{i*}P^*\|_2^p. 
$$
Again we can apply the inequality row-wise and obtain
$$
\|AP-AP^*\|_{p,2}^p = \sum_{i=1}^n \|A_{i*}P-A_{i*}P^*\|_{2}^p \le  \sum_{i=1}^n ( \|A_{i*}-A_{i*}P\|_{2}^p - \|A_{i*}-A_{i*}P^*\|_{2}^p) \le \epsilon \cdot \opt.
$$

Now we consider the final case of $1\le p <2$. Here we will make a case distinction.
The first case is that 
$\|A_{i*}P^* - A_{i*}P\|_2^p \le \frac{\epsilon}{4} \|A_{i*} - A_{i*}P\|_2^p$. Let $J$ be the set of indices
for which this inequality is satisfied. It follows by summing up over all rows in $J$ that
$$
\sum_{i\in J} \|A_{i*}P^* - A_{i*}P\|_2^p \le \frac{\epsilon}{4} \|A - AP\|_{p,2}^p \le \frac{\epsilon}{4} (1+\epsilon) \opt \le \frac{\epsilon}{2}\opt.
$$
For the remaining case we will use Claim \ref{claim:two} with $a= \|A_{i*}P-A_{i*}P^*\|_2$, $b=\|A_{i*}-A_{i*}P\|_2$ and $c=\|A_{i*}-A_{i*}P^*\|_2$.
We observe that $a^2 = b^2 - c^2$ and that $a^p \ge \frac{\epsilon}{4} b^p$ since we are in the second case. Furthermore, $b^p \ge c^p$ by the choices of
$P$ and $P^*$. Therefore, Claim \ref{claim:two} implies
$$
\|A_{i*}P-A_{i*}P^*\|_2^p  \le  10 (\epsilon/4)^{\frac{p-2}{p}} \big(\|A_{i*}-A_{i*}P\|_2^p - \|A_{i*}-A_{i*}P^*\|_2^p \big). 
$$
Applying the above inequality row wise we obtain
$$
\sum_{i \notin J}\|A_{i*}P^*-A_{i*}P\|_{2}^p \le  10 (\epsilon/4)^{\frac{p-2}{p}} \sum_{i\notin J}( \|A_{i*}-A_{i*}P\|_{2}^p - \|A_{i*}-A_{i*}P^*\|_{2}^p) \le \frac{\epsilon}{2} \opt .
$$ 
Summing up the two cases yields the lemma.
\end{proof}

\begin{remark}\label{rem:one}
We observe that in the proof we only used two properties of $S$. The first one is that $\|A-AP\|_{p,2}^p \le (1+\epsilon) \cdot \opt$ and
the second one is that $\|A-AP\|_{p,2}^p - \|A-AP^*\|_ {p,2}^p \le \epsilon^{\max\{\frac{2}{p},1\}} \opt/80$. Thus, any subspace that
satisfies these two properties will also satisfy the above lemma. We will use this later on when we discuss optimizing the running time
of our algorithm. 
\end{remark}

\subsection{Dimensionality reduction for sums of Euclidean distances}

We first consider the case of minimizing sum of distances. This case is technically less tedious and illustrates the underlying
ideas.

\begin{theorem}
\label{thm:one}
Let $1\ge \epsilon>0$.
Let $A \subseteq \REAL^{n\times d}$ be a matrix. Let $B \in \REAL^{n\times (d+1)}$ be the rank $O(k/\epsilon^2)$ matrix
output by algorithm {\sc DimensionalityReduction} with input 
parameters $A, n,d,k, p=1$ and $\epsilon/2$. Let $C \subseteq \REAL^d$ be an arbitrary non-empty set that is contained in a $k$-dimensional subspace. Let $A'$ and $B'$ be the matrices whose rows contain the closest points in the (closure of) $C$ wrt. the rows of $A$, $BI$, respectively.
Then we have
$$
\big| \|A-A'\|_{1,2} - \|B-B'I^T\|_{1,2} \big| \le \epsilon \|A-A'\|_{1,2},
$$
where $I \in \REAL^{(d+1) \times d}$ has diagonal entries $1$ and all other entries are $0$.
\end{theorem}

\begin{proof}
Let $S$ be the subspace as in Lemma \ref{lem:proj} and let $S^*$ be the span of $S$ and $C$ (if $S^*$ has less than $\ell+k$ dimensions, we can add arbitrary 
dimensions). Let $P, P^*$ be the corresponding orthogonal projection matrices. We know from Lemma \ref{lem:proj} that 
$$
\|AP - AP^*\|_{1,2} \le \frac{\epsilon}{2} \cdot \opt
$$
where $\opt$ is the cost of an optimal solution to the $k$-subspace problem with sum of distances.
By orthogonality, we can write 
$$
\|A_{i*} - A'_{i*}\|_2 = \sqrt{ \| A_{i*} - A_{i*}P^*\|_2^2 + \|A_{i*}P^* - A'_{i*}\|_2^2}.
$$
Furthermore, we have
$$
\|B_{i*} - B'_{i*} I^T\|_2 = \sqrt{  \|A_{i*}P - B'_{i*}\|_2^2+\|A_{i*} - A_{i*}P\|_2^2}. 
$$
Using Lemma \ref{lem:numbers} with $a= \| A_{i*} - A_{i*}P^*\|_2$, $b=\|A_{i*}P^* - A'_{i*}\|_2$, $e=\|A_{i*} - A_{i*}P\|_2$ and
$f= \|A_{i*}P - B'_{i*}\|_2$ we obtain that 
$$
|\|A_{i*} - A'_{i*}\|_2 - \|B_{i*} - B'_{i*} I^T\|_2| \le \big| \| A_{i*} - A_{i*}P^*\|_2-\|A_{i*} - A_{i*}P\|_2 \big| + 
\big| \|A_{i*}P^* - A'_{i*}\|_2 -\|A_{i*}P - B'_{i*}\|_2 \big|.
$$
Using the triangle inequality and the fact that $A'_{i*}, B'_{i*}$ is the closest point in the closure of $C$ to $A_{i*}P^*$ and
$A_{i*}P$, respectively, 
we obtain
$$
|\|A_{i*} - A'_{i*}\|_2 - \|B_{i*} - B'_{i*} I^T\|_2| \le 2 \cdot \|A_{i*}P - A_{i*}P^*\|_2
$$
Summing up over all rows and using $\|AP-AP^*\|_{1,2} \le \frac{\epsilon}{2} \cdot \opt$ together with the fact that the sum of distances to $C$
is at least $\opt$ we obtain the result.
\end{proof}

If $C$ is a $k$-dimensional linear subspace we can slightly simplify the statement of the above theorem.

\begin{corollary}\label{cor:useful}
Let $A \subseteq \REAL^{n\times d}$ be a matrix. Let $B \in \REAL^{n\times (d+1)}$ be the output of algorithm {\sc DimensionalityReduction} with input 
parameters $A, n,d,k, p=1$ and $\epsilon/2$. Let $P$ be an arbitrary rank $k$ orthogonal projection matrix.
Then we have
$$
\big| \|A-AP\|_{1,2} - \|B-BIPI^T\|_{1,2} \big| \le \epsilon \|A-AP\|_{1,2},
$$
where $I \in \REAL^{(d+1) \times d}$ has diagonal entries $1$ and all other entries are $0$.
\end{corollary}

\subsection{Dimensionality reduction for powers of Euclidean distances}

In order to obtain a dimensionality reduction for powers of Euclidean distances we follow the same approach as before.
The main challenge is that some calculations become more difficult as the triangle inequality is replaced by a relaxed
triangle inequality.

\begin{theorem}
\label{thm:two}
Let $p\ge 1$ be a constant.
Let $A \subseteq \REAL^{n\times d}$ be a matrix and $1 \ge \epsilon \ge 0$. Let $B \in \REAL^{n\times (d+1)}$ be the rank $O(k/\epsilon^{O(p)})$ matrix
output by algorithm {\sc DimensionalityReduction} with input parameters $A, n,d,k, p\ge 1$ and $\frac{\epsilon^{p+3}}3{(84p)^{2p}}$. Let $C \subseteq \REAL^d$ be an arbitrary non-empty set that is contained in a $k$-dimensional subspace. Let $A'$ and $B'$ be the matrices whose rows contain the closest points in the (closure of) $C$ wrt. the rows of $A$, $BI$, respectively.
Then we have
$$
\big| \|A-A'\|_{p,2}^p - \|B-B'I^T\|_{p,2}^p \big| \le \epsilon \|A-A'\|_{p,2}^p,
$$
where $I \in \REAL^{(d+1) \times d}$ has diagonal entries $1$ and all other entries are $0$.
\end{theorem}

\begin{proof}
Let $S$ be the subspace as in Lemma \ref{lem:proj} and let $S^*$ be the span of $S$ and $C$ (if $S^*$ has less than $\ell+k$ dimensions, we can add arbitrary 
dimensions). Let $P, P^*$ be the corresponding orthogonal projection matrices. We know from Lemma \ref{lem:proj} that 
$$
\|AP - AP^*\|_{p,2}^p \le \frac{\epsilon^{p+3}}{3(84p)^{2p}} \cdot \opt
$$
where $\opt$ is the cost of an optimal solution to the $k$-subspace problem with sum of powers of distances.
By orthogonality, we can write 
$$
\|A_{i*} - A'_{i*}\|_2 = \sqrt{ \| A_{i*} - A_{i*}P^*\|_2^2 + \|A_{i*}P^* - A'_{i*}\|_2^2}.
$$
Furthermore, we have
$$
\|B_{i*} - B'_{i*} I^T\|_2 = \sqrt{  \|A_{i*}P - B'_{i*}\|_2^2+\|A_{i*} - A_{i*}P\|_2^2}. 
$$
%
%
Now let us assume that $\|A_{i*} - A'_{i*}\|_2 \le \|B_{i*} - B'_{i*} I^T\|_2$ (the other case is analogous).
By the triangle inequality and the definition of $A'$ and $B'$ we have
\begin{eqnarray*}
\|A_{i*} P - B_{i*}'\|_2^2 & \le & \|A_{i*}P - A_{i*}'\|_2^2\\
&\le & \big( \|A_{i*}P -A_{i*}P^* \|_2 + \|A_{i*}P^* - A_{i*}'\|_2\big)^2\\
& \le & (1+\lambda^2) \cdot \|A_{i*} P^* - A_{i*}'\|_2^2 +(1+\frac{4}{\lambda^2})^2 \cdot \|A_{i*}P - A_{i*}P^*\|_2^2
\end{eqnarray*}
where the last inequality follows from Claim \ref{claim:five} with $\lambda^2$ replacing the $\epsilon$ there. 
We can write
$$
\|A_{i*} - A'_{i*}\|_2^p = (\sqrt{x})^p.
$$
where
$$
x= \|A_{i*} - A_{i*}P^*\|_2^2 +\|A_{i*}P^*-A_{i*}'\|_2^2 = \|A_{i*} - A_{i*}'\|_2^2. 
$$
Then we can write
$$
\|B_{i*} - B'_{i*} I^T\|_2^p\le (\sqrt{a + x})^p
$$
where 
\begin{eqnarray*}
a & = &  \lambda^2 \|A_{i*}P^* -A_{i*}'\|_2^2 + (1+\frac{4}{\lambda^2})^2 \|A_{i*}P - A_{i*}P^*\|_2^2 + \|A_{i*}P - A_{i*}P^*\|_2^2 \\
& \le & \lambda^2 \|A_{i*}P^* -A_{i*}'\|_2^2 + (2+\frac{4}{\lambda^2})^2 \|A_{i*}P - A_{i*}P^*\|_2^2.
\end{eqnarray*}
We use Claim \ref{claim:five} to obtain
$$
(\sqrt{a+x})^{p} \le (1+\frac{\epsilon}{3}) \cdot (\sqrt{x})^p + (1+\frac{6p}{\epsilon})^{p/2} \cdot (\sqrt{a})^p.
$$
Thus, using $\lambda = \epsilon^{1/2+1/p} / (21 p)$ we obtain
$$
(1+\frac{6p}{\epsilon})^{p/2} (\sqrt{a})^p \le \frac{\epsilon}{3} \|A_{i*} -A_{i*}'\|_2^p + \frac{(84p)^{2p}}{\epsilon^{p+2}}\|A_{i*}P - A_{i*}P^*\|_2^p\\
$$
Now it follows that 
$$
(\sqrt{a+x})^{p} - (\sqrt{x})^p \le \frac{\epsilon}{3} \cdot (\sqrt{x})^p +\frac{\epsilon}{3} \|A_{i*} -A_{i*}'\|_2^p + \frac{(84p)^{2p}}{\epsilon^{p+2}}\|A_{i*}P - A_{i*}P^*\|_2^p.\\
$$
The final result follows by summing up over all rows, replacing $(\sqrt{x})^p$ by $\|A_{i*}-A_{i*}'\|_2^p$, and plugging in the bound for $\|AP-AP^*\|_{p,2}^p$.

Now let us assume that $\|A_{i*} - A'_{i*}\|_2 > \|B_{i*} - B'_{i*} I^T\|_2$.
By the triangle inequality and the definition of $A'$ and $B'$ we have
\begin{eqnarray*}
\|A_{i*} P* - A_{i*}'\|_2^2 & \le & \|A_{i*}P* - B_{i*}'\|_2^2\\
&\le & \big( \|A_{i*}P* -A_{i*}P \|_2 + \|A_{i*}P - B_{i*}'\|_2\big)^2\\
& \le & (1+\lambda^2) \cdot \|A_{i*} P - B_{i*}'\|_2^2 +(1+\frac{4}{\lambda^2})^2 \cdot \|A_{i*}P - A_{i*}P^*\|_2^2
\end{eqnarray*}
where the last inequality follows from Claim \ref{claim:five} with $\lambda^2$ replacing the $\epsilon$ there. 
We can write
$$
\|B_{i*} - B'_{i*}I^T\|_2^p = (\sqrt{x})^p.
$$
where
$$
x= \|A_{i*} - A_{i*}P\|_2^2 +\|A_{i*}P-B_{i*}'\|_2^2. 
$$
Then we can write
$$
\|A_{i*} - A'_{i*}\|_2^p \le (\sqrt{a + x})^p
$$
where 
\begin{eqnarray*}
a & = &  \lambda^2 \|A_{i*}P -B_{i*}'\|_2^2 + (1+\frac{4}{\lambda^2})^2 \|A_{i*}P - A_{i*}P^*\|_2^2 \\
\end{eqnarray*}

We use Claim \ref{claim:five} to obtain
$$
(\sqrt{a+x})^{p} \le (1+\frac{\epsilon}{3}) \cdot (\sqrt{x})^p + (1+\frac{6p}{\epsilon})^{p/2} \cdot (\sqrt{a})^p.
$$
Thus, using $\lambda = \epsilon^{1/2+1/p} / (21 p)$ we obtain
\begin{eqnarray*}
(1+\frac{6p}{\epsilon})^{p/2} (\sqrt{a})^p &\le& \frac{\epsilon}{3} \|B_{i*} -B_{i*}' I^T\|_2^p + \frac{(84p)^{2p}}{\epsilon^{p+2}}\|A_{i*}P - A_{i*}P^*\|_2^p\\
& \le & \frac{\epsilon}{3} \|A_{i*} -A_{i*}' I^T\|_2^p + \frac{(84p)^{2p}}{\epsilon^{p+2}}\|A_{i*}P - A_{i*}P^*\|_2^p\\
\end{eqnarray*}
Now it follows that 
$$
(\sqrt{a+x})^{p} - (\sqrt{x})^p \le \frac{\epsilon}{3} \cdot (\sqrt{x})^p +\frac{\epsilon}{3} \|A_{i*} -A_{i*}'\|_2^p + \frac{(84p)^{2p}}{\epsilon^{p+2}}\|A_{i*}P - A_{i*}P^*\|_2^p.\\
$$
The final result follows by summing up over all rows, replacing $(\sqrt{x})^p$ by $\|A_{i*}-A_{i*}'\|_2^p$, and plugging in the bound for $\|AP-AP^*\|_{p,2}^p$.

\end{proof}

\subsection{Optimizing the Running Time}\label{sec:alternative}
We first give an alternative algorithm to our {\sc DimensionalityReduction} algorithm. This algorithm can be implemented using a black box call to an algorithm for finding low dimensional subspaces approximately minimizing the $\|\cdot\|_{p,2}^p$ norm.
\begin{algorithm}[H]
	\caption{Dimensionality Reduction Algorithm II}
	\begin{algorithmic}[1]
		\Procedure{DimensionalityReductionII}{$A,n,d, k, \epsilon, p$}
                 \State $\tau = \Theta(\epsilon^{\max(2/p,1)})$. 
                 \State Choose a random $i^* \in \{1, 2, \ldots, 10/\tau\}$
	         \State Let $S$ be an $i^*k$-dimensional subspace $E$ 
                        with $$\|A(I-P_E)\|_{p,2}^p \leq (1+\Theta(\epsilon^{\max(2/p, 1)}))\min_{\textrm{rank-}ik E'}\|A(I-P_{E'})\|_{p,2}^p.$$
                        Such a space $S$ can be found by Theorem 1 \cite{cw15} with the $k$ there equal 
                        to our $i^*k$, and the $\epsilon$ there can be set to our $\tau$, if $p \in [1,2)$. 
                        The success probability is at least $9/10$. For $p > 2$, one can use the algorithm in
                        \cite{DV07} together with \cite{SV07}. 
                \State  For $i = 1, \ldots, n$, output a $(1 \pm \epsilon)$-approximation to $\|A_i(I-P_S)\|_2$. 
                        These $n$ values can be found by solving $n$ regression problems each with probability $1-1/n^2$, using the
                        regression algorithm of \cite{cw13}. See Lemma \ref{lem:regression2} below.
		\EndProcedure
\end{algorithmic}
\end{algorithm}

\begin{lemma}\label{lem:subspace2}
With probability at least $4/5$, {\sc DimensionalityReductionII} finds a $O(k/\epsilon^{\max(2/p,1)})$-dimensional subspace $S$ for which for all $k$-dimensional spaces $W$, $\|A(I-P_{S})\|_{p,2}^p - \|A(I-P_{S \cup W})\|_{p,2}^p \leq \Theta(\epsilon^{\max(2/p,1)}) \opt$. 
Further, for $p \in [1,2)$, finding such an $S$ can be done in 
$O(\nnz(A) + (n+d)\poly(k/\epsilon) + \exp(\poly(k/\epsilon))$ time, and for $p > 2$, 
can be found in $O(\nnz(A)\poly(k/\epsilon) +(n+d) \poly(k/\epsilon)+ \exp(\poly(k/\epsilon))$ time. 
\end{lemma}
\begin{proof}
For $p\in[1,2)$ we condition on the event that the algorithm of \cite{cw15} for computing a $(1+\tau)$-approximation $S$ succeeds, 
which holds with probability at least $9/10$. For $p>2$ we condition on the event that the algorithm of \cite{DV07} for computing a $(1+\tau/10)$-approximation $S$ succeeds, which holds with probability at least $9/20$. In the case of $p>2$ the algorithm of \cite{DV07} only returns a $\poly(k/\epsilon)$-dimensional subspace.
This will suffice to obtain a coreset for parameter $\tau/10$ of size $\poly(k/\epsilon)$ in the stated running time using the construction 
described in this paper. In order to obtain an $i^*k$-dimensional subspace in the claimed running time, we rotate this coreset so that it is
spanned by $\poly(k/\epsilon)$ standard basis vectors and run the algorithm of \cite{SV07} with parameter $\tau/(20k)$
and sufficiently many repetitions so that the best subspace returned is with probability $9/20$ a $(1+\tau)$-approximation for the 
original point set (and where we evaluate the quality of the subspace on the coreset and we can amplify the success probability with the median trick).
Then we reverse the rotation. The above procedure can be implemented in time 
$\nnz(A) \poly(k/\epsilon) +(n+d) \poly(k/\epsilon)+ \exp(\poly(k/\epsilon))$ time.

The algorithm requires time $O(\nnz(A) + (n+d)\poly(k/\epsilon) + \exp(\poly(k/\epsilon))$ for $p \in [1,2)$ and $O(\nnz(A)\poly(k/\epsilon) + \exp(\poly(k/\epsilon)))$ for $p > 2$. 
For each $j \in \{1, 2, \ldots, 10/\tau+1\}$, 
let $V^j$ be the optimal $jk$-dimensional subspace, and consider a telescoping sum:
\begin{eqnarray*}
\opt - \|A(I-P_{V^{10/\tau+1}})\|_{p,2}^p & \geq & \|A(I-P_{V^1})\|_{p,2}^p - \|A(I-P_{V^{20/\tau+1}})\|_{p,2}^p\\
& = & \sum_{i=1}^{10/\tau+1} (\|A(I-P_{V^{i-1}})\|_{p,2}^p - \|A(I-P_{V^i})\|_{p,2}^p)\\
& \geq & 0.
\end{eqnarray*}
There are $10/\tau$ summands in the telescoping sum, and they sum up to at most $\opt$, so a $9/10$-fraction
of them must be at most $\tau \opt$. Let $i^*$ be the index sampled by the algorithm. Then with probability at least $9/10$,
we have $\|A(I-P_{V^{i^*}})\|_{p,2}^p - \|A(I-P_{V^{i^+1}})\|_{p,2}^p \leq \tau \opt$, and let us condition on this event. 

Now, $\tilde{V}^{i^*} \cup W$ is an $(i^*+1)k$-dimensional subspace, and so we have
$\|A(I-P_{\tilde{V}^{i^*} \cup W})\|_{p,2}^p \geq \|A(I-P_{V^{i^*+1}})\|_{p,2}^p$. Also, by the guarantee of $\tilde{V}^{i^*}$, we have
$\|A(I-P_{\tilde{V}^{i^*}})\|_{p,2}^p \leq (1+\tau)\|A(I-P_{V^{i^*}})\|_{p,2}^p \leq (1+\tau)(\|A(I-P_{V^{i^*}})\|_{p,2}^p + \tau \opt )$,
where $\opt$ is the cost of the best $k$-dimensional subspace. Consequently, for any $k$-dimensional subspace $W$, 
\begin{eqnarray*}
\|A(I-P_{\tilde{V}^{i^*}})\|_{p,2}^p - \|A(I-P_{\tilde{V}^{i^*} \cup W})\|_{p,2}^p & \leq & 
(1+\tau)(\|A(I-P_{V^{i^*}})\|_{p,2}^p + \tau \opt) - \|A(I-P_{V^{i^*+1}})\|_{p,2}^p\\
& \leq & O(\tau) \opt,
\end{eqnarray*}
where we used that $\|A(I-P_{V^{i^*}})\|_{1,2} - \|A(I-P_{V^{i^*+1}})\|_{1,2} \leq \tau \opt$ for our choice of $i^*$, 
and also that $\|A(I-P_{V^{i^*}})\|_{p,2}^p \leq \opt$.

Note that the overall success probability is at least $1-1/10 -1/10$ in the first case and $1-1/10-1/20 -1/20 = 4/5$ in the second, and the claimed
running time follows from \cite{cw15},\cite{DV07} and \cite{SV07}. 
\end{proof}

\begin{lemma}\label{lem:regression2}
With probability at least $1-1/n$, {\sc DimensionalityReductionII} outputs a $(1 \pm \epsilon)$-approximation to 
$\|A_i(I-P_S)\|_2$ simultaneously for every $i \in [n]$. 
Further, this can be done in 
$O(\nnz(A)\log n + (n+d)\poly(k/\epsilon) \log n)$ time. 
\end{lemma}
\begin{proof}
We set up $n$ regression problems, the $i$-th being:
$\min_x \|A_i - xV^T\|_2$, where $P_S = VV^T$ is the orthogonal projection onto $S$, and $V^T$ is an orthonormal basis for $S$ which can be computed from $S$ in $d \cdot \poly(k/\epsilon)$ time. Using input sparsity time algorithms for regression \cite{cw13,mm13,nn13}, if we choose a CountSketch matrix $S$ with 
$\poly(k/\epsilon)$ rows, then with probability at least $9/10$, 
$\|A_i S - xV^TS\|_2 = (1 \pm \epsilon)\|A_i - xV^T\|_2$ for all $x$. We can compute $AS$ in $\nnz(A)$ time and $V^TS$ in $d \cdot \poly(k/\epsilon)$ time, at which point we can solve the $n$ regression problems in $n \cdot \poly(k/\epsilon)$ time, so $\nnz(A) + (n+d)\poly(k/\epsilon)$ total time. We repeat the entire procedure $O(\log n)$ times, and take the median of our $O(\log n)$ estimates, for each $i \in \{1, 2, \ldots, n\}$. This amplifies the success probability
to $1-1/n^2$ for each $i \in [n]$, and a union bound gives a $1-1/n$ probability bound simultaneously for all $i \in [n]$. The total time is $O(\nnz(A) \log n + (n+d)\poly(k/\epsilon) \log n)$. 
%
\end{proof}

We also need the following lemma stating that the approximations returned by Lemma \ref{lem:regression2} suffice. 
\begin{lemma}\label{lem:guarantee2}
Let $B \in \REAL^{n\times (d+1)}$ be a matrix for which for any rank $k$ orthogonal projection matrix $P \in \REAL^{d\times d}$ we have
\begin{eqnarray}\label{eqn:guaranteeChanged}
\big| \|A-AP\|_{p,2}^p - \|B-BIPI^T\|_{p,2}^p \big| \le \epsilon \|A-AP\|_{p,2}^p.
\end{eqnarray}
Suppose we replace the last column $v$ of $B$ with a vector $v'$ for which $v_i = (1 \pm \epsilon)v'_i$ for all $i \in [n]$. Then 
(\ref{eqn:guaranteeChanged}) continues to hold with $\epsilon$ replaced with $O(\epsilon)$. 
\end{lemma}
\begin{proof}
Notice that this operation changes $\|B-BIPI^T\|_{p,2}^p$ by at most a $(1+O(\epsilon))$ factor for constant $p$, since each row changes by at most this factor. 
Letting $B'$ denote the new matrix, we have
\begin{eqnarray*}
|\|A-AP\|_{p,2}^p - \|B'-B'IPI^T\|_{p,2}^p| & = & |\|A-AP\|_{p,2}^p - \|B-BIPI^T\|_{p,2}^p \pm \epsilon \|B-BIPI^T\|_{p,2}^p|\\
& \leq & \epsilon \|A-AP\|_{p,2}^p + \epsilon  \|B-BIPI^T\|_{p,2}^p\\
& \leq& \epsilon \|A-AP\|_{p,2}^p + \epsilon (\|A-AP\|_{p,2}^p + \epsilon \|A-AP\|_{p,2}^p)\\
& \leq & (2\epsilon + \epsilon^2)\|A-AP\|_{p,2}^p,
\end{eqnarray*}
and rescaling $\epsilon$ by a constant factor gives the desired guarantee.
\end{proof}

We cannot afford to compute the projection of $A$ onto $S$, as this could take longer than $\tilde{O}(\nnz(A))$ time. Fortunately, we show in
Section \ref{sec:subspace} that we do not need to compute this. For the $k$-median problem we need the following additional lemma to 
approximate matrix $B$ (see also Remark \ref{rem:one}). We fix an error in equation (2) from an earlier version (see also 
\cite{FKW21}).

\begin{lemma}
\label{lemma:tildeB}
Let $S$ be the subspace guaranteed by Lemma \ref{lem:subspace2}. Given $S$ we can compute in time
$O(\nnz(A)\log n + (n+d)\poly(k/\epsilon))$ a matrix $\tilde{B}$ of rank $O(k/\epsilon^{\max(2/p,1)})$ such that
with probability at least $9/10$ we have for every set $C$ contained in a $k$-dimensional subspace 
$$|\|B-B'I\|_{p,2}^p - \|\tilde{B}-\tilde{B}'I\|_{p,2}^p| \leq \epsilon \|A-A'\|_{p,2}^p.$$
Here $B'$ and $\tilde B'$ are matrices that contain in the $i$-th row 
in the first $d$ coordinates the point from (the closure of) $C$ that is closest
to the i-th row of $BI$ and $\tilde BI$ respectively, and have the $d+1$-coordinate $0$. 
\end{lemma}
\begin{proof}
Let $P = VV^T$ be the orthogonal projection onto $S$. We run the algorithm of 
\cite{ANW14} (see also Section 2.3 of \cite{w14}) with parameter $\epsilon^2/3$
 which gives us with probability at least $9/10$ in time $O(\nnz(A)\log n + (n+d)\poly(k/\epsilon))$, 
simultaneously for each $i \in [n]$, a vector $\tilde{A}_i = X_i V^T$ 
for which $\|\tilde{A}_i - A_i\|_2 \leq (1+\epsilon^2/3)\|A_iP - A_i\|_2$ and so 
$\|\tilde{A}_i - A_i\|_2^2 \leq (1+\epsilon^2/3)^2 \|A_iP-A_i\|_2^2 \leq (1+\epsilon^2) \|A_iP-A_i\|_2^2$.
By Pythagorean theorem we have, 
\begin{eqnarray}\label{eqn:regression}
\|\tilde{A}_i - A_iP\|_2^2 = \|A_i - \tilde{A}_i\|_2^2 - \|A_i - A_iP\|_2^2 \leq \epsilon^2 \|A_i-A_iP\|_2^2.
\end{eqnarray}
Taking the root implies $\|\tilde{A}_i - A_iP\|_2 \leq \epsilon \|A_i - A_iP\|_2$.
By the triangle inequality, $|\dist(A_iP, C) - \dist(\tilde{A}_i, C)| \leq \dist(A_iP, \tilde{A}_i)$, and combining with (\ref{eqn:regression}),
this implies 
\begin{eqnarray}\label{eqn:regression2}
|\dist(A_iP, C) - \dist(\tilde{A}_i, C)|\leq \epsilon \cdot \dist(A_i, A_iP),
\end{eqnarray}
for each $i \in [n]$. 
The $i$-th row $\tilde{B}_i$ of $\tilde{B}$ consists of the coordinates of $\tilde{A}_i$ followed by the single coordinate
of value $(1 \pm \epsilon) \dist(A_i, S)$. Consequently, for each $i$:
\begin{eqnarray}\label{eqn:regression3}
\|\tilde{B}_i - (\tilde{B}'I)_i\|_2^p = (\dist(\tilde{A}_i, C)^2 + (1 \pm \epsilon)\dist(A_i, S)^2)^{p/2},
\end{eqnarray}
and so,
\begin{eqnarray*}
|\tilde{B}_i - (\tilde{B}'I)_i\|_2^p & = & (\dist(\tilde{A}_i, C)^2 + (1 \pm \epsilon) \dist(A_i, S)^2)^{p/2}\\
& = & ((\dist(A_iP, C) \pm \epsilon \dist(A_i, A_iP))^2 + (1 \pm \epsilon)\dist(A_i, S)^2)^{p/2}\\
& = & (\dist(A_iP, C)^2 \pm O(\epsilon) \dist(A_iP, C)d(A_i, A_iP) \\
& & \pm O(\epsilon^2)d(A_i, A_iP)^2 + (1 \pm \epsilon)\dist(A_i, S)^2)^{p/2}\\
& = & (\|B_i - (B'I)_i\|_2^2 \pm O(\epsilon)(\dist(A_i, S)^2 \\
&& + \dist(A_iP, C)d(A_i, S) \pm O(\epsilon \cdot \dist(A_i, S)^2)))^{p/2}\\
& = & (\|B_i - (B'I)_i\|_2^2 \pm O(\epsilon)(\dist(A_i, S)^2 + \dist(A_iP, C)^2))^{p/2}\\
& = & (1 \pm O(\epsilon))\|B_i - (B'I)_i\|_2^p
\end{eqnarray*}
where the first equality uses the definition of $B$, the second equality uses (\ref{eqn:regression2}), 
the third equality just expands the square, the fourth equality uses (\ref{eqn:regression3}) and that $\dist(A_i, A_iP) = \dist(A_i, S)$, 
the penultimate equality absorbs the previous equality in the asymptotic notation $O(\epsilon)$, and the final equality uses
that $\|B_i - (B'I)_i\|_2^2 = \dist(A_iP, C)^2 + \dist(A_i, S)^2$. 

Hence, $||\tilde{B}_i - (\tilde{B}'I)_i\|_2^2 - \|B_i - (B'I)_i\|_2^p| = O(\epsilon)\|B_i - (B'I)_i\|_2^p$, 
and so 
\begin{eqnarray}\label{eqn:relative}
|\|B-B'I\|_{p,2}^p - \|\tilde{B}-\tilde{B}'I\|_{p,2}^p| = O(\epsilon)\|B-(B'I)\|_{p,2}^p.
\end{eqnarray}
Now, for a given $i$, 
$\|B_i - (B'I)_i\|_2^2 = \dist(A_iP, C)^2 + \dist(A_i, S)^2$. Notice that $\dist(A_iP, C) \leq \dist(A_i, C) + \dist(A_iP, A_i) = \dist(A_i, C) + \dist(S, A_i)$, 
and consequently $\|B_i - (B'I)_i\|_2^2 = O(\dist(A_i, C)^2 + \dist(A_i, S)^2)$, and plugging into (\ref{eqn:relative}),
\begin{eqnarray*}
|\|B-B'I\|_{p,2}^p - \|\tilde{B}-\tilde{B}'I\|_{p,2}^p| & = & O(\epsilon)\sum_{i=1}^n (\dist(A_i, C)^2 + \dist(A_i, S)^2)^{p/2}\\
& = & O(\epsilon)(\sum_{i=1}^n \dist(A_i, C)^p + \sum_{i=1}^n \dist(A_i, S)^p)\\
& = & O(\epsilon)(\sum_{i=1}^n \dist(A_i, C)^p)\\
& = & O(\epsilon)\|A-A'\|_{p,2}^p,
\end{eqnarray*}
where the first equality follows by plugging into (\ref{eqn:relative}), the second equality
follows from $(\dist(A_i, C)^2 + \dist(A_i, S)^2)^{p/2} \leq (2\max(\dist(A_i, C)^2, \dist(A_i, S)^2))^{p/2} 
\leq 2^{p/2} \max(\dist(A_i, C)^p, \dist(A_i, S)^p) \leq 2^{p/2} (\dist(A_i, C)^p + \dist(A_i, S)^p)$, 
the third equality follows from $\sum_i \dist(A_i, S)^p \leq (1+\epsilon) \sum_i \dist(A_i, C)^p$ by definition of $S$
and the fact that $C$ is contained in a $k$-dimensional subspace, and the final equality just uses the definition
of $\|A-A'\|_{p,2}^p$. 
\end{proof}

\section{Coresets for Subspace Approximation}\label{sec:subspace}

Plugging in the guarantee of Lemma \ref{lem:subspace2} into Remark \ref{rem:one} shows how to obtain an $n \times (d+1)$ matrix $B$ for which
\begin{eqnarray}\label{eqn:guarantee}
|\|A-AP\|_{p,2}^p - \|B-BIPI^T\|_{p,2}^p| \le \epsilon \|A-AP\|_{p,2}^p,
\end{eqnarray}
for all rank-$k$ orthogonal projection matrices $P$. Lemma \ref{lem:subspace2} shows how to efficiently find the $S$
for which $B = AP_S$, and Lemma \ref{lem:regression2} shows how to efficiently find an approximate vector $v$ of $(d+1)$-st coordinates of $B$. Further,
after scaling $\epsilon$ by a constant factor, Lemma \ref{lem:guarantee2} shows that 
(\ref{eqn:guarantee}) continues to hold with the approximate vector $v$ of $(d+1)$-st coordinates of $B$ furnished by Lemma \ref{lem:regression2}. 

The main issue is that we cannot afford to compute the projection of $A$ onto $S$ to form the first $d$ columns of $B$. 
Although the matrix $B$ is $n \times (d+1)$, it has rank $poly(k/\epsilon)$. Thinking of its rows as 
$n$ points in $\mathbb{R}^{d+1}$, we would like to find a {\it weighted subset} $TB$ of these $n$ points
so that $\|TB-TBIPI^T\|_{1,2} = (1 \pm \epsilon)\|B-BIPI^T\|_{1,2}$ for all rank-$k$ orthogonal projection matrices $P$.
Here, $T$ is called a sampling and rescaling matrix, and our goal will be to find such a $T$ for which each row of 
$T$ contains a single non-zero non-negative entry, corresponding to the sampled row of $B$, rescaled
by a non-negative value. 

Recall that $B$ has a particular form, namely, $B = [AP_S, v]$, where $P_S$ is the orthogonal projection
onto the $\poly(k/\epsilon)$-dimensional subspace found by our {\sc DimensionalityReduction} or 
{\sc DimensionalityReductionII} algorithm, and 
$v_i = (1 \pm \epsilon) \|A_{i*} - A'_{i*}\|_{2}$,
where $A' = AP_S$ for $i = 1, 2, \ldots, n$. Since $S$ is $\poly(k/\epsilon)$-dimensional, we can write $P_S = UU^T$, where
$U$ is a $d \times \poly(k/\epsilon)$ matrix with orthonormal columns. 

Before showing how to find a sampling and rescaling matrix $T$, we first need the following theorem due to Dvoretsky.
\begin{fact}(Dvoretsky's Theorem, see Variations and Extensions on p.30 of \cite{m13})\label{thm:dvoretsky}
Let $t \geq C d \log(1/\epsilon)/\epsilon^2$ for a sufficiently large constant $C > 0$, and suppose $G$ is a $d \times t$
matrix of i.i.d. $N(0,1/t)$ random variables, where $N(0, 1/t^2)$ denotes a normal random variable with mean $0$ and variance $1/t$. 
Then with probability at least $99/100$, simultaneously
for all $x \in \mathbb{R}^d$, $\|x\|_2 = (1 \pm \epsilon)\|xG\|_1$. In particular, there exists such a matrix $G$
for which this property holds for all $x \in \mathbb{R}^d$. For $p > 1$, there is a $d \times t'$ matrix
$G$ of i.i.d. normal random variables, suitably scaled, with $t' = (d/\epsilon)^{O(p)}$ for which
for all $x \in \mathbb{R}^d$, $\|x\|_2 = (1 \pm \epsilon)\|xG\|_p$. 
\end{fact}

\begin{lemma}(Sampling Lemma) \label{lem:sampling}
Given $S$, in $n \cdot \poly(k (\log n)/\epsilon)$ time it is possible to find a sampling and rescaling matrix $T$ with 
$O(\textrm{rank}(S) \log(\textrm{rank(S)}/\epsilon^2)$ rows
for which for all rank-$k$ orthogonal projection matrices $P$, 
$$\|TB-TBIPI^T\|_{1,2} = (1 \pm \epsilon)\|B-BIPI^T\|_{1,2}.$$
Instantiating $S$ with the output of our {\sc DimensionalityReduction} algorithm, $T$ would have
$O(k \log(k/\epsilon) / \epsilon^4)$ rows. Instantiating $S$ with the output of our {\sc DimensionalityReductionII} algorithm, $T$
would have $\poly(k/\epsilon)$ rows. 

For constant $p > 1$, it is possible to find a sampling and rescaling matrix $T$ with 
$\poly(\textrm{rank}(S)/\epsilon)$ rows for which for all rank-$k$ orthogonal projection matrices $P$,
$$\|TB-TBIPI^T\|_{p,2}^p = (1 \pm \epsilon)\|B-BIPI^T\|_{p,2}^p.$$
Instantiating $S$ with the output of our {\sc DimensionalityReduction} or {\sc DimensionalityReductionII} algorithms, 
$T$ would have $\poly(k/\epsilon)$ rows.
\end{lemma}
\begin{proof}
Let $t$ be as in Fact \ref{thm:dvoretsky}, and fix the $d \times t$ matrix $G$ of that fact.
Applying the guarantee of Fact \ref{thm:dvoretsky} to each row of $B-BIPI^T$, 
\begin{eqnarray}\label{eqn:proof1}
\|B-BIPI^T\|_{p,2}^p = (1 \pm \Theta(\epsilon)) \|BG-BIPI^TG\|_{p,p}^p, 
\end{eqnarray}
where for a matrix $C$, 
$\|C\|_{p,p}^p$ denotes the sum of $p$-th powers of absolute values of its entries. 

We next apply Theorem 1 of \cite{cp15}, which shows how, given a matrix $C$ with $f$ columns,
to find a sampling and rescaling matrix $T$ with $O(f \epsilon^{-2} \log f)$ rows for $p = 1$,
and $(f \epsilon^{-1})^{O(p)}$ rows for $p > 1$, 
for which with high probability, simultaneously
for all $x \in \mathbb{R}^t$, $\|TCx\|_p = (1 \pm \epsilon) \|Cx\|_p$. Further, the time 
to find $T$ is $O(\log \log n)$ calls to computing $2$-approximate statistical leverage
scores of matrices of the form $WC$, where $W$ is a non-negative diagonal matrix. Using the algorithm
of Theorem 29 of \cite{cw13}, $T$ can be computed in $\tilde{O}(\nnz(C))$ time. 

Instantiating the matrix $C$ of the previous paragraph with the $n \times O(k/\epsilon^2)$ 
matrix $AU$, where recall $P_S = UU^T$ and $U$ has $k/\epsilon^2$ columns if $S$ is the output of our 
{\sc DimensionalityReduction} algorithm, it follows that $\|TAUx\|_1 = (1\pm \epsilon)\|AUx\|_1$
for all $x \in \mathbb{R}^{O(k/\epsilon^2)}$, and $T$ has ${O(k \log(k/\epsilon)/\epsilon^4)}$ rows. If $S$ is the output
of our {\sc DimensionalityReductionII} algorithm or $p > 1$, then $T$ has $\poly(k/\epsilon)$ rows. The overall time
is $n \cdot \poly(k(\log n)/\epsilon)$. 

Consequently, 
\begin{eqnarray}\label{eqn:proof2}
\|TBG-TBIPI^TG\|_{p,p}^p = (1 \pm \epsilon)\|BG-BIPI^TG\|_{p,p}^p, 
\end{eqnarray}
since each column
of $BG$ is in the column span of $AU$. Again applying the guarantee of Fact \ref{thm:dvoretsky}
to each row of $TBG-TBIPI^TG$, we have 
\begin{eqnarray}\label{eqn:proof3}
\|TBG-TBIPI^TG\|_{p,p}^p = (1 \pm \epsilon)\|TB-TBIPI^T\|_{p,2}^p.
\end{eqnarray}
Combining (\ref{eqn:proof1}), (\ref{eqn:proof2}), and (\ref{eqn:proof3}), we have
$$\|B-BIPI^T\|_{p,2}^p = (1 \pm O(\epsilon))\|TB-TBIPI^T\|_{p,2}^p,$$
and the guarantee of the lemma follows by rescaling $\epsilon$ by a constant factor. 
\end{proof}

\begin{theorem}(Strong Coreset for Subspace Approximation)\label{thm:subspaceApprox}
For $p = 1$, it is possible to find a matrix 
$TB \in \mathbb{R}^{O(k (\log k)/\epsilon^4) \times d+1}$ for which for all rank-$k$ orthogonal projection 
matrices $P$, 
\begin{eqnarray}\label{eqn:guaranteeZap}
|\|A-AP\|_{p,2} - \|TB-TBIPI^T\|_{p,2}| \le \epsilon \|A-AP\|_{p,2}.
\end{eqnarray}
Further, in $\tilde{O}(\nnz(A) + (n+d)\poly(k/\epsilon) + \exp(\poly(k/\epsilon))$ time, 
it is possible to find a matrix
$TB \in \mathbb{R}^{\poly(k/\epsilon) \times d+1}$ satisfying (\ref{eqn:guaranteeZap}) for $p \in [1,2)$ 
for all rank-$k$ orthogonal projection matrices $P$. 

Finally, in $\nnz(A)\poly(k/\epsilon) + \exp(\poly(k/\epsilon))$ time, it is possible to find a matrix
$TB \in \mathbb{R}^{\poly(k/\epsilon) \times d+1}$ satisfying (\ref{eqn:guaranteeZap}) for $p > 2$ 
for all rank-$k$ orthogonal projection matrices $P$. 
\end{theorem}
\begin{proof}
We start by proving the structural part of the theorem, and then address the running time. 

Let $B$ be the output of {\sc DimensionalityReduction}, which has property (\ref{eqn:guarantee}). 
As described above, we can assume {\sc DimensionalityReduction}  
produces $B = [AP_S, v]$, where $P_S$ and $v$ are described above. Further, we
can assume $AP_S$ is given in factored form $(AU) U^T$ for $P_S = UU^T$. 

By Lemma \ref{lem:sampling}, we can find a sampling and rescaling
matrix $T$ 
for which for all rank-$k$ orthogonal projection matrices $P$, 
$\|TB-TBIPI^T\|_{p,2}^p = (1 \pm \epsilon)\|B-BIPI^T\|_{p,2}^p$, so 
\begin{eqnarray*}
|\|A-AP\|_{p,2}^p - \|TB-TBIPI^T\|_{p,2}^p| & = & |\|A-AP\|_{p,2}^p - \|B-BIPI^T\|_{p,2}^p| \pm \epsilon \|B-BIPI^T\|_{p,2}^p\\
& \leq & \epsilon \|A-AP\|_{p,2}^p + \epsilon  \|B-BIPI^T\|_{p,2}^p\\
& \leq& \epsilon \|A-AP\|_{p,2}^p + \epsilon (\|A-AP\|_{p,2}^p + \epsilon \|A-AP\|_{p,2}^p )\\
& \leq & (2\epsilon + \epsilon^2)\|A-AP\|_{p,2}^p,
\end{eqnarray*}
and rescaling $\epsilon$ by a constant factor gives the desired guarantee. We note that By Lemma \ref{lem:sampling},
$T$ will have $O(k (\log k)/\epsilon^4)$ rows for $p = 1$ if we run {\sc DimensionalityReduction}.

For an efficient algorithm, we instead run {\sc DimensionalityReductionII} 
in $\tilde{O}(\nnz(A) + (n+d)\poly(k/\epsilon) + \exp(\poly(k/\epsilon))$ time for $p \in [1,2)$, 
and in $\nnz(A)\poly(k/\epsilon) + \exp(\poly(k/\epsilon))$ time for $p > 2$, 
to obtain $S$, which is $\poly(k/\epsilon)$-dimensional,
and in $n \cdot \poly(k/\epsilon)$ time we can compute $U$, where $P_S = UU^T$. The correctness is given by Lemma \ref{lem:subspace2}. We
also compute the $(d+1)$-st column of $B$ in $\tilde{O}(\nnz(A) + (n+d)\poly(k/\epsilon))$ time. By Lemma \ref{lem:sampling},
in $n \cdot \poly(k (\log n)/\epsilon)$ time we can find a sampling and rescaling matrix $T$. 
Finally, $T$ selects $\poly(k/\epsilon)$ rows of $A$, and for each we compute its projection onto $S$,
taking $d \cdot \poly(k/\epsilon)$ time in total. We also output the corresponding entry of $v$. We thus obtain our coreset
$TB$ in the stated running times. 
\end{proof}

%
%
%

\section{Coresets for $k$-Median}\label{sec:kmedian}

We can combine our dimensionality reduction with coreset computations. We first compute the matrix $B$ using
our dimensionality reduction. Recall that this matrix has rank $O(k/\epsilon^2)$. We can therefore compute an orthogonal
basis for the span of $B$ and add $k$ arbitrary dimensions. Then we apply an arbitrary coreset construction in the resulting
space. We will only be interested in the space spanned by the first $d$ dimensions of $B$ (recall that the last dimension is only needed
to ``adjust" the distances). Since the Euclidean distance does not change under orthogonal transformations we can rotate any set of centers to
our subspace without changing the distances. Therefore, the coreset will be a coreset for the whole input space. 

\begin{theorem}
Let $1 \ge \epsilon >0$.
Given a matrix $A \in \REAL^{n\times d}$. We can compute in time $\tilde{O}(\nnz(A) + (n+d)\poly(k/\epsilon) + \exp(\poly(k/\epsilon)))$ a matrix 
$T \in \REAL^{s \times (d+1)}$, $s=O(\frac{k^2 \log k}{\epsilon^4})$, and non-negative weights $w_1,\dots, w_s$ such that 
with probability at least $3/5$ for every set $C$ of $k$ centers we have
$$
\big| \|A-A^C\|_{1,2} - \sum_{i=1}^s w_i \|T_{i*} - T_{i*}^C \|_{2} \big| \le \epsilon \|A-A^C\|_{1,2}
$$
Here $A^C$ contains for each row of $A$ the nearest center of $C$, $T^C$ contains for each row of $T$ 
the nearest center of $C$ with respect to $T I^T$.
\end{theorem}

\begin{proof}
We use Lemma \ref{lem:subspace2} (plugging it into Remark \ref{rem:one}) to compute with probability at least $9/10$
in time $O(\nnz(A) + (n+d)\poly(k/\epsilon) + \exp(\poly(k/\epsilon))$ the subspace $S$
of rank $O(k/\epsilon^2)$ using parameter $\epsilon/10$. We then apply Lemma \ref{lemma:tildeB} to obtain with probability at least $4/5$ matrix $\tilde B$ of rank 
$O(k/\epsilon^2)$ from matrix $B$ defined by $S$ in time $O(\nnz(A)\log n + (n+d)\poly(k/\epsilon))$. Next we apply the coreset construction of \cite{FL11} 
or \cite{BFL16} to obtain in time $\tilde O(n \poly(k \log (1/\delta) /\epsilon))$ (recall that the dimension of our space is $\poly(k/\epsilon)$ 
a coreset $T^*$ of size $O(k^2 \log k/ \epsilon^4)$, where we set the error probability $\delta =1/10$. Finally we compute for each coreset point the corresponding coordinates 
in the original space in $O(d \poly(k/\epsilon))$ time to obtain the matrix $T$.

We get the following guarantees. Theorem \ref{thm:one} implies
$$
\big|\|\|A-A'\|_{1,2} - \|B-B'I\|_{1,2}\| \big|\le \epsilon \|A-A'\|_{1,2}.
$$                                  
By Lemma \ref{lemma:tildeB} we know that
$$
|\|B-B'I\|_{1,2} - \|\tilde{B}-(\tilde{B}I)^C\|_{1,2}| \leq \epsilon \|A-A'\|_{1,2}
$$
For the coreset $S$ we know that
$$
\big| \|\tilde B -(\tilde BI)^C\|_{1,2} - \sum_{i=1}^{|T|} w_i \cdot \|T_{i*}- T_{i*}^C\|_2 \big| \le \epsilon \|\tilde B - (\tilde BI)^C\|_{1,2}.
$$
Combining the first two statements gives
\begin{eqnarray*}
\big|\|\|A-A'\|_{1,2} - \sum_{i=1}^{|T|} w_i \cdot \|T_{i*}- T_{i*}^C \|_2 \big| &\leq & 2\epsilon \|A-A'\|_{1,2} + \epsilon \|\tilde B - (\tilde BI)^C\|_{1,2} \\
& \le & 2 \epsilon \|A-A'\|_{1,2} + \epsilon ( (1+2\epsilon) \epsilon \|A-A'\|_{1,2})\\
& \le & 5 \epsilon  \|A-A'\|_{1,2}
\end{eqnarray*}

Rescaling $\epsilon$ gives the result.
\end{proof}

\section*{Acknowledgment}
Christian Sohler acknowledges the support of the German Science Foundation (DFG) Collaborative Research Center SFB 876 "Providing Information by Resource-Constrained Analysis", project A2.
David Woodruff acknowledges support in part by an Office of Naval Research (ONR) grant N00014-18-1-2562.

We thank Lingxiao Huang for pointing out an error in the proof of Lemma 14 in an earlier version of this paper.


\bibliographystyle{plain}
\bibliography{main}

\begin{thebibliography}{10}

\bibitem{ANW14}
Haim Avron, Huy~L. Nguyen, and David~P. Woodruff.
\newblock Subspace embeddings for the polynomial kernel.
\newblock In {\em Advances in Neural Information Processing Systems 27: Annual
  Conference on Neural Information Processing Systems 2014, December 8-13 2014,
  Montreal, Quebec, Canada}, pages 2258--2266, 2014.

\bibitem{BHI02}
Mihai Badoiu, Sariel Har{-}Peled, and Piotr Indyk.
\newblock Approximate clustering via core-sets.
\newblock In {\em Proceedings on 34th Annual {ACM} Symposium on Theory of
  Computing, May 19-21, 2002, Montr{\'{e}}al, Qu{\'{e}}bec, Canada}, pages
  250--257, 2002.

\bibitem{BZMD15}
Christos Boutsidis, Anastasios Zouzias, Michael~W. Mahoney, and Petros Drineas.
\newblock Randomized dimensionality reduction for k-means clustering.
\newblock {\em {IEEE} Trans. Information Theory}, 61(2):1045--1062, 2015.

\bibitem{BFL16}
Vladimir Braverman, Dan Feldman, and Harry Lang.
\newblock New frameworks for offline and streaming coreset constructions.
\newblock {\em CoRR}, abs/1612.00889, 2016.

\bibitem{C09}
Ke~Chen.
\newblock On coresets for k-median and k-means clustering in metric and
  euclidean spaces and their applications.
\newblock {\em {SIAM} J. Comput.}, 39(3):923--947, 2009.

\bibitem{cw13}
Kenneth~L. Clarkson and David~P. Woodruff.
\newblock Low rank approximation and regression in input sparsity time.
\newblock In {\em Symposium on Theory of Computing Conference, STOC'13, Palo
  Alto, CA, USA, June 1-4, 2013}, pages 81--90, 2013.

\bibitem{cw15}
Kenneth~L. Clarkson and David~P. Woodruff.
\newblock Input sparsity and hardness for robust subspace approximation.
\newblock In {\em {IEEE} 56th Annual Symposium on Foundations of Computer
  Science, {FOCS} 2015, Berkeley, CA, USA, 17-20 October, 2015}, pages
  310--329, 2015.

\bibitem{CEMMP15}
Michael~B. Cohen, Sam Elder, Cameron Musco, Christopher Musco, and Madalina
  Persu.
\newblock Dimensionality reduction for k-means clustering and low rank
  approximation.
\newblock In {\em Proceedings of the Forty-Seventh Annual {ACM} on Symposium on
  Theory of Computing, {STOC} 2015, Portland, OR, USA, June 14-17, 2015}, pages
  163--172, 2015.

\bibitem{cp15}
Michael~B. Cohen and Richard Peng.
\newblock L\({}_{\mbox{p}}\) row sampling by lewis weights.
\newblock In {\em Proceedings of the Forty-Seventh Annual {ACM} on Symposium on
  Theory of Computing, {STOC} 2015, Portland, OR, USA, June 14-17, 2015}, pages
  183--192, 2015.

\bibitem{DRVW06}
Amit Deshpande, Luis Rademacher, Santosh Vempala, and Grant Wang.
\newblock Matrix approximation and projective clustering via volume sampling.
\newblock {\em Theory of Computing}, 2(12):225--247, 2006.

\bibitem{DV07}
Amit Deshpande and Kasturi~R. Varadarajan.
\newblock Sampling-based dimension reduction for subspace approximation.
\newblock In {\em Proceedings of the 39th Annual {ACM} Symposium on Theory of
  Computing, San Diego, California, USA, June 11-13, 2007}, pages 641--650,
  2007.

\bibitem{FL11}
Dan Feldman and Michael Langberg.
\newblock A unified framework for approximating and clustering data.
\newblock In {\em Proceedings of the 43rd {ACM} Symposium on Theory of
  Computing, {STOC} 2011, San Jose, CA, USA, 6-8 June 2011}, pages 569--578,
  2011.

\bibitem{FMSW10}
Dan Feldman, Morteza Monemizadeh, Christian Sohler, and David~P. Woodruff.
\newblock Coresets and sketches for high dimensional subspace approximation
  problems.
\newblock In {\em Proceedings of the Twenty-First Annual {ACM-SIAM} Symposium
  on Discrete Algorithms, {SODA} 2010, Austin, Texas, USA, January 17-19,
  2010}, pages 630--649, 2010.

\bibitem{FSS13}
Dan Feldman, Melanie Schmidt, and Christian Sohler.
\newblock Turning big data into tiny data: Constant-size coresets for
  \emph{k}-means, {PCA} and projective clustering.
\newblock In {\em Proceedings of the Twenty-Fourth Annual {ACM-SIAM} Symposium
  on Discrete Algorithms, {SODA} 2013, New Orleans, Louisiana, USA, January
  6-8, 2013}, pages 1434--1453, 2013.

\bibitem{FS12}
Dan Feldman and Leonard~J. Schulman.
\newblock Data reduction for weighted and outlier-resistant clustering.
\newblock In {\em Proceedings of the Twenty-Third Annual {ACM-SIAM} Symposium
  on Discrete Algorithms, {SODA} 2012, Kyoto, Japan, January 17-19, 2012},
  pages 1343--1354, 2012.

\bibitem{FKW21}
Zhili Feng, Praneeth Kacham, and David~P. Woodruff.
\newblock Dimensionality reduction for the sum-of-distances metric.
\newblock In Marina Meila and Tong Zhang, editors, {\em Proceedings of the 38th
  International Conference on Machine Learning, {ICML} 2021, 18-24 July 2021,
  Virtual Event}, volume 139 of {\em Proceedings of Machine Learning Research},
  pages 3220--3229. {PMLR}, 2021.

\bibitem{FS05}
Gereon Frahling and Christian Sohler.
\newblock Coresets in dynamic geometric data streams.
\newblock In {\em Proceedings of the 37th Annual {ACM} Symposium on Theory of
  Computing, Baltimore, MD, USA, May 22-24, 2005}, pages 209--217, 2005.

\bibitem{FS06}
Gereon Frahling and Christian Sohler.
\newblock A fast k-means implementation using coresets.
\newblock In {\em Proceedings of the 22nd {ACM} Symposium on Computational
  Geometry, Sedona, Arizona, USA, June 5-7, 2006}, pages 135--143, 2006.

\bibitem{HK07}
Sariel Har{-}Peled and Akash Kushal.
\newblock Smaller coresets for k-median and k-means clustering.
\newblock {\em Discrete {\&} Computational Geometry}, 37(1):3--19, 2007.

\bibitem{HM04}
Sariel Har{-}Peled and Soham Mazumdar.
\newblock On coresets for k-means and k-median clustering.
\newblock In {\em Proceedings of the 36th Annual {ACM} Symposium on Theory of
  Computing, Chicago, IL, USA, June 13-16, 2004}, pages 291--300, 2004.

\bibitem{LS10}
Michael Langberg and Leonard~J. Schulman.
\newblock Universal epsilon-approximators for integrals.
\newblock In {\em Proceedings of the Twenty-First Annual {ACM-SIAM} Symposium
  on Discrete Algorithms, {SODA} 2010, Austin, Texas, USA, January 17-19,
  2010}, pages 598--607, 2010.

\bibitem{m13}
Jir{\i} Matou{\v{s}}ek.
\newblock Lecture notes on metric embeddings.
\newblock Technical report, 2013.

\bibitem{mm13}
Xiangrui Meng and Michael~W. Mahoney.
\newblock Low-distortion subspace embeddings in input-sparsity time and
  applications to robust linear regression.
\newblock In {\em Symposium on Theory of Computing Conference, STOC'13, Palo
  Alto, CA, USA, June 1-4, 2013}, pages 91--100, 2013.

\bibitem{nn13}
Jelani Nelson and Huy~L. Nguyen.
\newblock {OSNAP:} faster numerical linear algebra algorithms via sparser
  subspace embeddings.
\newblock In {\em 54th Annual {IEEE} Symposium on Foundations of Computer
  Science, {FOCS} 2013, 26-29 October, 2013, Berkeley, CA, {USA}}, pages
  117--126, 2013.

\bibitem{SV07}
Nariankadu~D. Shyamalkumar and Kasturi~R. Varadarajan.
\newblock Efficient subspace approximation algorithms.
\newblock In {\em Proceedings of the Eighteenth Annual {ACM-SIAM} Symposium on
  Discrete Algorithms, {SODA} 2007, New Orleans, Louisiana, USA, January 7-9,
  2007}, pages 532--540, 2007.

\bibitem{VX12}
Kasturi~R. Varadarajan and Xin Xiao.
\newblock On the sensitivity of shape fitting problems.
\newblock In {\em {IARCS} Annual Conference on Foundations of Software
  Technology and Theoretical Computer Science, {FSTTCS} 2012, December 15-17,
  2012, Hyderabad, India}, pages 486--497, 2012.

\bibitem{w14}
David~P. Woodruff.
\newblock Sketching as a tool for numerical linear algebra.
\newblock {\em Foundations and Trends in Theoretical Computer Science},
  10(1-2):1--157, 2014.

\end{thebibliography}
\newpage
\appendix
\section{Appendix}
\begin{proof}( of Claim \ref{claim:one})
By our assumptions we have
$$
a^p +c^p = (a^2)^{p/2} + (c^2)^{p/2} \le (a^2 + c^2)^{p/2} = b^p.
$$
Subtracting $c^p$ yields the claim. 
\end{proof}

\begin{proof}( of Claim \ref{claim:two})
We observe that for $c=0$ the result is trivial, so we may assume $c \not= 0$.
We know that $a^2+c^2= b^2$ and so $(a^2+c^2)^{p/2} = b^p$.
 We write $a= \delta c$ for some $\delta > 0$ (the case $\delta=0$ implies $a=b=c=0$ and the result follows immediately). 
For the case $\delta <1$ we obtain
\begin{eqnarray*}
(a^2 + c^2)^{p/2} & = & (\delta^2 c^2+c^2)^{p/2} \\
&= & (1+\delta^2)^{p/2} c^p\\
& \ge & (1+\delta^2)^{1/2} c^p\\
& \ge & (1 + \delta^2/3) c^p\\
& = & \frac{1}{3} \delta^{2-p} \cdot a^p + c^p\\
& \ge & \frac{1}{3} \epsilon^{\frac{2-p}{p}} \cdot a^p + c^p\\
\end{eqnarray*}
where the second inequality follows from $(1+\delta^2/3)^2 \le 1+\delta^2$. The last inequality follows from 
from $\delta^p c^p = a^p \ge \epsilon b^p \ge \epsilon c^p$ (since c>0) and hence $\delta^p \ge \epsilon$.

Plugging this into $(a^2+c^2)^{p/2} = b^p$ gives
$$
\frac{1}{3} \epsilon^{\frac{2-p}{p}} \cdot a^p + c^p \le b^p
$$
which implies the claim for $\delta <1$ by subtracting $c^p$ and multiplying with $3 \epsilon^{\frac{p-2}{p}}$.

If $1\le \delta \le 2$ we can write 
\begin{eqnarray*}
(a^2 + c^2)^{p/2} & = & (\delta^2 c^2+c^2)^{p/2} \\
&= & (1+\delta^2)^{p/2} c^p\\
& \ge & \sqrt{2} c^p\\
& \ge & \frac{2}{5} c^p + c^p\\
&\ge &  \frac{1}{10} \delta^p c^p +c^p\\
& = & \frac{1}{10} a^p +c^p\\
\end{eqnarray*}
Plugging this into $(a^2+c^2)^{p/2} = b^p$ gives
$$
\frac{1}{10} \cdot a^p + c^p \le b^p
$$
which implies the claim since $\epsilon^{\frac{p-2}{p}} \ge 1$.

Finally, if $\delta >2$ then
\begin{eqnarray*}
(a^2 + c^2)^{p/2} & = & (\delta^2 c^2+c^2)^{p/2} \\
&= & (1+\delta^2)^{p/2} c^p\\
& \ge & \delta^p c^p\\
& \ge & (1+\delta^p/2) c^p\\
& = &  \frac{1}{2} \delta^p c^p +c^p\\
& =  & \frac{1}{2} a^p +c^p\\
\end{eqnarray*}
and the claim follows analogous to the above.

\end{proof}

\begin{proof}( of Claim \ref{claim:four})
We can assume wlog. $a\ge b$ and then this is equivalent to showing $(a+x)^p - a^p \le (b+x)^p - b^p$. When $x = 0$ both sides are $0$. 
The derivative, as a function of $t$, of $f_c(t) := (c+t)^p - c^p$, for a fixed $c$, is $p/(c+t)^{1-p},$ so the derivative $df_a(t)/dt$ 
is less than or equal to $df_b(t)/dt$ for every $t \ge 0$, and by the fundamental theorem of calculus, 
$$
(a+x)^p - a^p = \int_{i = 0}^x (df_a(t)/dt ) dt \le \int_{t = 0}^x (df_b(t)/dt) dt = (b+x)^p-b^p.
$$
\end{proof}

\begin{proof}( of Claim \ref{claim:five})
Observe that $(1+\epsilon/(2p))^p 
\le 1+\epsilon$ follows from the monotonicity of the logarithm and $\ln ((1+\epsilon/(2p))^p) \le p \cdot \epsilon/(2p) \le \epsilon - \epsilon^2/2
\le \ln (1+\epsilon)$ by the Taylor expansion of $\ln (1+x)$ for $0 < \epsilon <1$ (for $\epsilon=1$ the Claim follows immediately). 
If $b \le \frac{\epsilon}{2p} \cdot a$ then $(a+b)^p \le (1+\frac{\epsilon}{2p})^p \cdot a^p \le (1+\epsilon) a^p$.
Otherwise, $b > \frac{\epsilon}{2p} \cdot a$. In this case $(a+b)^p \le (1+\frac{2p}{\epsilon})^p b^p$.
\end{proof}

\begin{proof}( of Lemma \ref{lem:numbers})
We have by the triangle inequality and $2$-norm to $1$-norm relationship, 
$$|\sqrt{a^2+b^2} - \sqrt{f^2+g^2}| = |\|(a,b)\|_2 - \|(f,g)\|_2| \leq \|(a-f, b-g)\|_2 \leq \|(a-f, b-g)\|_1 = |a-f| + |b-g|.$$
\end{proof}

\end{document}